\documentclass{article}

\usepackage{bussproofs}
\usepackage{amsmath}
\usepackage{amsfonts}
\usepackage{amssymb}
\usepackage{amsthm}
\usepackage{stmaryrd}
\usepackage{verbatim}
\usepackage{amscd}
\usepackage[dvips]{epsfig}
\usepackage{color}
\usepackage{rotating}
\usepackage{wrapfig}
\usepackage{subfigure}
\usepackage{etex}
\input xy
\xyoption{all}
\usepackage{cmll} 
	\newcommand{\lone}{1}
	\newcommand{\ltens}{\otimes}
	\newcommand{\lbot}{\bot}
	
	\newcommand{\lpar}{\parr}
	
\usepackage{mathbbol}

\usepackage[UglyObsolete]{diagrams}

\usepackage{authblk}

\def\restriction#1#2{\mathchoice
              {\setbox1\hbox{${\displaystyle #1}_{\scriptstyle #2}$}
              \restrictionaux{#1}{#2}}
              {\setbox1\hbox{${\textstyle #1}_{\scriptstyle #2}$}
              \restrictionaux{#1}{#2}}
              {\setbox1\hbox{${\scriptstyle #1}_{\scriptscriptstyle #2}$}
              \restrictionaux{#1}{#2}}
              {\setbox1\hbox{${\scriptscriptstyle #1}_{\scriptscriptstyle #2}$}
              \restrictionaux{#1}{#2}}}
\def\restrictionaux#1#2{{#1\,\smash{\vrule height .8\ht1 depth .85\dp1}}_{\,#2}}
\def\N{{{\rm I}\!{\rm N}}}

\theoremstyle{plain}
\newtheorem{theorem}{Theorem}
\theoremstyle{plain}
\newtheorem{definition}[theorem]{Definition}
\theoremstyle{plain}
\newtheorem{fact}[theorem]{Fact}
\theoremstyle{plain}
\newtheorem{corollary}[theorem]{Corollary}
\theoremstyle{plain}
\newtheorem{lemma}[theorem]{Lemma}
\theoremstyle{plain}
\newtheorem{remark}{Remark}
\theoremstyle{plain}
\newtheorem{example}[theorem]{Example}

\theoremstyle{plain}

\theoremstyle{plain}
\newtheorem{proposition}[theorem]{Proposition}

\newcommand{\integers}[0]{\N}
\newcommand{\finitesequences}[1]{{#1}^{< \omega}}
\newcommand{\card}[1]{\textit{Card}(#1)}
\newcommand{\strong}[1]{\textit{strong}(#1)}
\newcommand{\nonerasingstrong}[1]{\textit{strong}_{\nonerasing}(#1)}
\newcommand{\sequence}[1]{\langle #1 \rangle}
\newcommand{\exhaustive}[1]{{#1}^\textit{ex}}
\newcommand{\sm}[1]{\llbracket #1 \rrbracket}
\newcommand{\smbis}[1]{\Lparen #1 \Rparen}
\newcommand{\sizepoint}[1]{s(#1)}
\newcommand{\sizepointbis}[1]{s_\mathcal{W}(#1)}

\newcommand{\sizeexperiment}[1]{s_{\sm{}}(#1)}

\newcommand{\sizebisinf}[1]{{s_{\mathcal{W}}}_\textit{inf}(#1)}
\newcommand{\sizeexperimentbis}[1]{s_{\smbis{}}(#1)}
\newcommand{\sizenet}[1]{\lVert #1 \rVert}
\newcommand{\cutnets}[4]{(#1 \vert #2)_{#3, #4}}

\newcommand{\compl}[1]{\sharp #1}
\newcommand{\experimentbis}[2]{\Lparen \Rparen\textrm{-experiment }#1\textrm{ of }#2}
\newcommand{\Cut}[1]{\textit{Cut}(#1)}
\newcommand{\result}[1]{\vert #1 \vert}

\newcommand{\finitemultisets}[1]{\mathcal{M}_\textit{fin}(#1)}
\newcommand{\finitesubsets}[1]{\mathcal{P}_\textit{fin}(#1)}
\newcommand{\subsets}[1]{\mathcal{P}(#1)}
\newcommand{\multiset}[1]{[#1]}
\newcommand{\atomic}[1]{#1_\textit{At}}
\newcommand{\weakeningsofexperiment}[1]{\mathcal{W}(#1)}
\newcommand{\depth}[1]{\textit{depth}(#1)}
\newcommand{\rank}[1]{\textit{rank}(#1)}
\newcommand{\Supp}[1]{\textit{Supp}(#1)}

\newcommand{\seq}[1]{(#1)}
\newcommand{\length}[1]{\textit{length}(#1)}
\newcommand{\ground}[1]{\textit{ground}(#1)}
\newcommand{\onecut}{\rightsquigarrow}
\newcommand{\erasing}{e}
\newcommand{\nonerasing}{\neg e}
\newcommand{\stratnonerasing}{{{(\nonerasing)}_{s}}}

\newcommand{\antistraterasing}{{\erasing}_{as}}
\newcommand{\oneerasing}{{\onecut}_{\erasing}}
\newcommand{\onestratnonerasing}{{\onecut}_\stratnonerasing}

\newcommand{\onenonerasing}{{\onecut}_{\nonerasing}}
\newcommand{\oneantistraterasing}{{\onecut}_{\antistraterasing}}

\newcommand{\cutred}{{\onecut}^\ast}
\newcommand{\stratnonerasingred}{{\onecut}_{\stratnonerasing}^\ast}
\newcommand{\erasingred}{{\oneerasing}^\ast}
\newcommand{\nonerasingred}{{\onenonerasing}^\ast}

\newcommand{\gnet}{g-structure}
\newcommand{\grnet}{ground-structure}
\newcommand{\scalefact}{0.67}

\newenvironment{minilist}{\begin{list}{$\bullet$}{
  \setlength{\parsep}{0pt}
  \setlength{\topsep}{-10pt}
  \setlength{\itemsep}{-\parsep}
  \setlength{\labelsep}{0.4em}
  \setlength{\leftmargin}{1.3em}}}{\end{list}}
\newenvironment{minienum}{\begin{enumerate}{
  \setlength{\parsep}{0pt}
  \setlength{\topsep}{-10pt}
  \setlength{\itemsep}{-\parsep}
  \setlength{\labelsep}{0.4em}
  \setlength{\leftmargin}{1.3em}}}{\end{enumerate}}

\newarrow{DashTo}{}{dash}{}{dash}{>}
\newarrow{DashtrTo}{}{dash}{}{dash}{>>}

\title{A semantic account of strong normalization in Linear Logic}

\author[1]{Daniel de Carvalho}
\author[2]{Lorenzo Tortora de Falco}

\affil[1]{Datalogisk Institut, K\o benhavns Universitet}
\affil[2]{Dipartimento di Matematica e Fisica, Roma~III}

\begin{document}

\maketitle

\begin{abstract}
We prove that given two cut-free nets of linear logic, by means of their relational interpretations one can: 1) first determine whether or not the net obtained by cutting the two nets is strongly normalizable 2) then (in case it is strongly normalizable) compute the maximum length of the reduction sequences starting from that net. 
\end{abstract}

\section{Introduction}

Linear Logic (LL,~\cite{ll}) originated from the coherent model of typed $\lambda$-calculus: the category of coherent spaces and linear maps was ``hidden'' behind the category of coherent spaces and stable maps. It then turned out that the coherence relation was not necessary to interpret linear logic proofs (proof-nets), and this remark led to the so-called multiset based relational model of LL: the interpretation of proof-nets in the category $\mathbf{Rel}$ of sets and relations. Since then, many efforts have been done to understand to which extent the relational interpretation of a proof-net is nothing but a different representation of the proof itself: in Girard's original paper (\cite{ll}), with every proof-net was associated the set of ``results of experiments'' of the proof-net, a set proven to be invariant with respect to cut elimination. Later on these ``results'' have been represented as nets themselves, and through Taylor's expansion a proof-net can been represented as an infinite linear combination of nets (see~\cite{EhrhardRegnier:DiffNets} and~\cite{EhrhardRegnier:UnifTaylor}). On the other hand, we proved in~\cite{MR2926280} that (in the absence of weakenings) one can always recover, from the relational interpretation of a cut-free proof-net, the proof-net itself.

This paper establishes another tight link between the relational model and LL proof-nets. We follow the approach to the semantics of bounded time complexity consisting in measuring by semantic means the execution of any program, regardless of its computational complexity. The aim is to compare different computational behaviors and to learn something afterwards on the very nature of bounded time complexity. Following this approach and inspired by~\cite{bohmKrivine}, in \cite{phddecarvalho,Carvalhoexecution} one of the authors of the present paper could compute the execution time of an untyped $\lambda$-term from its interpretation in the Kleisli category of the comonad associated with the finite multisets functor on the category of sets and relations. Such an interpretation is the same as the interpretation of the net encoding the $ \lambda$-term in the multiset based relational model of linear logic. The execution time is measured there in terms of elementary steps of the so-called Krivine machine. Also, \cite{phddecarvalho,Carvalhoexecution} give a precise relation between an intersection types system introduced in~\cite{CDV} and experiments in the multiset based relational model. Experiments are  a tool introduced by Girard in~\cite{ll} allowing to compute the interpretation of proofs pointwise. An experiment corresponds to a type derivation and the result of an experiment corresponds to a type. This same approach was applied in~\cite{CarvPagTdF10} to LL to show how it is possible to compute the number of steps of cut elimination by semantic means (notice that the measure being now the number of cut elimination steps, here is a first difference with \cite{phddecarvalho,Carvalhoexecution} where Krivine's machine was used to measure execution time). The results of~\cite{CarvPagTdF10} are presented in the framework of proof-nets, that we call nets in this paper: if $\pi'$ is a net obtained by applying some steps of cut elimination to $\pi$, the main property of any model is that the interpretation $\sm{\pi}$ of $\pi$ is the same as the interpretation $\sm{\pi'}$ of $\pi'$, so that from $\sm{\pi}$ it is clearly impossible to determine the number of steps leading from $\pi$ to $\pi'$. Nevertheless, in~\cite{CarvPagTdF10} it is shown that if $\pi_1$ and $\pi_2$ are two cut-free nets connected by means of a cut-link, one can answer the two following questions by only referring to the interpretations $\sm{\pi_1}$ and $\sm{\pi_2}$ in the relational model:
\begin{itemize}
\item\label{uno1}
is it the case that the net obtained by cutting $\pi_{1}$ and $\pi_{2}$ is weakly normalizable?
\item\label{due2}
if the answer to the previous question is positive,
what is the number of cut reduction steps
leading from the net with cut to a cut-free one?
\end{itemize}

In the present paper, still by only referring to the interpretations $\sm{\pi_1}$ and $\sm{\pi_2}$ in the relational model, we answer the two following variants of the previous questions:
\begin{enumerate}
\item\label{uno1bis}
is it the case that the net obtained by cutting $\pi_{1}$ and $\pi_{2}$ is strongly normalizable?
\item\label{due2bis}
if the answer to the previous question is positive, what is the maximum length (i.e.\ the number of cut reduction steps) of the reduction sequences starting from the net obtained by cutting $\pi_{1}$ and $\pi_{2}$?
\end{enumerate}

Despite the fact that the new questions are just little variations on the old ones, the answers \emph{are not} variants of the old ones, and require the development of new tools (see for example the new $\smbis{}$-interpretation of Definition~\ref{def:experiment}). The first question makes sense only in an untyped framework (in the typed case, cut elimination is strongly normalizing, see~\cite{ll,phddanos,SNLL10} and...Subsection~\ref{subsection:conservation}!), and we thus study in Section~\ref{sect:nets} nets and their stratified reduction in an untyped framework. Subsection~\ref{subsect:netsDefinition} mainly recalls definitions and notations coming from~\cite{CarvPagTdF10}, while in Subsection~\ref{subsect:netsReduction}, we prove two syntactic results that will be used in the sequel: 1) Proposition~\ref{proposition:SN=SNnonerasing} reduces strong normalization to ``non erasing'' strong normalization (and will be used in Section~\ref{sect:SN}), and 2) Proposition~\ref{proposition:nonerasing-erasing} shows that when a net is strongly normalizable there exists a ``canonical'' reduction sequence of maximum length, consisting first of ``non erasing stratified'' steps and then of ``erasing antistratified'' steps (and will be used in Section~\ref{sect:SNquantitative}).

In Section~\ref{sect:experiments}, we introduce the standard notion of experiment (called $\sm{}$-experiment in this paper) leading to the usual interpretation (called $\sm{}$-interpretation in this paper) of a net in the category of sets and relations (the multiset based relational model of linear logic). In the same Definition~\ref{def:experiment}, we introduce $\smbis{}$-experiments, leading to the $\smbis{}$-interpretation of nets: the main difference between $\sm{}$-experiments and $\smbis{}$-experiments is the behavior w.r.t.\ weakening links. And indeed, the main difference between weak and strong normalization lies in the fact that to study the latter property we cannot ``forget pieces of proofs'' (and this is actually what the usual $\sm{}$-interpretation does by assigning the empty multiset as label to the conclusion of weakening links). The newly defined $\smbis{}$-interpretation \emph{does not} yield a model of linear logic: it is invariant only w.r.t.\ \emph{non erasing} reduction steps (Proposition~\ref{prop:invariancesmbis}).

In Section~\ref{sect:SN}, we point out an intrinsic difference between the semantic characterization of strong normalization and the one of weak normalization proven in~\cite{CarvPagTdF10} (here Theorem~\ref{th:qualitatifWN}): there exist nets $\pi$ and $\pi'$ such that $\sm{\pi}=\sm{\pi'}$ and $\pi$ is strongly normalizing while $\pi'$ is not, which clearly shows that there is no hope (in the general case) to extract the information on the strong normalizability of a net from its $\sm{}$-interpretation (Remark~\ref{rem:NoSNsem}). We then prove that in case $\pi$ is a cut-free net, its $\smbis{}$-interpretation $\smbis{\pi}$ can be computed from its ``good old'' $\sm{}$-interpretation $\sm{\pi}$ (Proposition~\ref{prop:SembisFromSem}). This implies that to answer Questions~\ref{uno1bis} and~\ref{due2bis} by only referring to the interpretations $\sm{\pi_1}$ and $\sm{\pi_2}$ in the ``good old'' relational model of linear logic, we are allowed to use the newly defined $\smbis{}$-interpretations $\smbis{\pi_1}$ and $\smbis{\pi_2}$. We then accurately adapt the notion of size of an $\sm{}$-experiment of the relational model to $\smbis{}$-experiments, in order to obtain a variant of the ``Key Lemma'' (actually Lemmata 17 and 20) of~\cite{CarvPagTdF10}: Lemma~\ref{lemma : key-lemma : strat} measures the difference between the size of (suitable) experiments of a net and the size of (suitable) experiments of any of its one step reducts. We can thus answer Question~\ref{uno1bis} (Corollary~\ref{corollary : cut strongly normalizable}).\\ 
Our qualitative results of Subsection~\ref{subsect:SNsem} allow to give a new proof of the so called ``Conservation Theorem'' (here Theorem~\ref{theorem: conservation}) for Multiplicative Exponential Linear Logic ($MELL$). Such a result is a crucial step in the traditional proof of strong normalization for Linear Logic (\cite{ll,phddanos,SNLL10}) and it is usually proven using confluence (\cite{phddanos,SNLL10}): our semantic approach does not rely on confluence and yields thus a proof of strong normalization for $MELL$ which does not use confluence (Corollary~\ref{corollary:SN} of Subsection~\ref{subsection:conservation}).

In Section~\ref{sect:SNquantitative}, we answer Question~\ref{due2bis}: thanks to Proposition~\ref{proposition:nonerasing-erasing} it is enough from $\sm{\pi_1}$ and $\sm{\pi_2}$ to predict the length of a ``canonical'' reduction sequence, and by Proposition~\ref{prop:SembisFromSem} we can substitute $\smbis{\pi_{1}}$ and $\smbis{\pi_{2}}$ for $\sm{\pi_1}$ and $\sm{\pi_2}$. We first measure the length of the longest ``non erasing stratified'' reduction sequence, by means of the size of (suitable) experiments, and we then shift to the size of results of $\smbis{}$-experiments, that is elements of the $\smbis{}$-interpretation. We then measure the length of the longest ``erasing antistratified'' reduction sequence starting from a ``non erasing normal'' net, relating this length to the number of (erasing) cuts of the net, and counting this number using the $\smbis{}$-interpretation. The precise answer to Question~\ref{due2bis} is Theorem~\ref{theorem:exactSN}. We end the section by giving a concrete example (Example~\ref{example:LessThanInjectivity}), showing also that only a little part of $\sm{\pi_1}$ and $\sm{\pi_2}$ is used in Theorem~\ref{theorem:exactSN} to compute the maximum length of the reduction sequences starting from the net obtained by cutting $\pi_{1}$ and $\pi_{2}$.\\
In a parallel non communicating work (\cite{DBLP:conf/fossacs/BernadetL11,bernadetleng11b,bernadetlengrand13}), a semantic bound of the number of $\beta$-reductions of a given $\lambda$-term is given. We briefly point out some differences and similarities between the two approaches in Remark~\ref{rem:IntersectTpesRex} and in the conclusion of the paper; it would probably worth comparing more precisely our result with those papers in future work.

\subsection*{Notations}
For a set $X$, $\subsets{X}$ denotes the set of the subsets of $X$, $\finitesubsets{X}$ denotes the set of the finite subsets of $X$ and $\finitemultisets{X}$ denotes the set of finite multisets of elements of $X$. 
The number of elements of a finite set $X$ is denoted by $\card{X}$. 
As usual, a finite multiset of elements of $X$ is a function with domain $X$ and codomain the set $\integers$ of natural numbers; when $m\in\finitemultisets{X}$, we denote by $\Supp{m}$ the subset of $X$ having as elements those $a\in X$ such that $m(a)>0$, and more generally for any $x\in X$, the integer $m(x)$ is sometimes called the multiplicity of $x$ in $m$.
We write $a+b$ for the sum of the two finite multisets $a$ and $b$, and for a finite multiset $m$ of elements of the set $X$ we denote by $\card{m}$ the integer $\Sigma_{x\in \Supp{m}} m(x)$.

Given any set $X$, we denote by $\finitesequences{X}$ the set of finite sequences of elements of $X$, and by $\mathbf{x}$ a generic element of $\finitesequences{X}$. For example, a sequence $\seq{c_1,\dots,c_n}$ may be denoted simply by $\mathbf{c}$.

\section{Nets and their normalization}\label{sect:nets}

In this section, we introduce nets and their cut elimination in an untyped framework (Subsection~\ref{subsect:netsDefinition}), mainly following~\cite{CarvPagTdF10}. We then study normalization of these nets (Subsection~\ref{subsect:netsReduction}): the two main results that will be used in the sequel are 1) a net is strongly normalizable iff every \emph{non erasing} reduction sequence starting from it is finite (Proposition~\ref{proposition:SN=SNnonerasing}) and 2) whenever a net $\pi$ is strongly normalizing, there exist ``canonical'' reduction sequences of maximum length starting from $\pi$ that first reduce stratified non erasing cuts and then erasing cuts (Proposition~\ref{proposition:nonerasing-erasing}).

\subsection{Nets}\label{subsect:netsDefinition}

The theory of proof-nets has rather changed since the introduction of this crucial concept of linear logic in~\cite{ll}: we choose here the syntax of~\cite{CarvPagTdF10}, where we already discussed such a choice. Let us just recall here that untyped nets in our sense have been first introduced in~\cite{lics06} in order to encode polytime computations (inspired by the ``light'' untyped $\lambda$-calculus of~\cite{phdterui}). One of the novelties of the untyped \emph{classical} framework of~\cite{lics06} w.r.t. the intuitionistic framework of~\cite{phdterui} is the presence of \emph{clashes}, that is cuts which cannot be reduced (see Definition~\ref{def:clash} and Figure~\ref{fig:clashes}). Following~\cite{hilbert} we consider $?$-links with $n\geq 0$ premises (these links are often represented by a tree of contractions and weakenings), while our $\flat$-node is a way to represent dereliction: these choices allowed in~\cite{CarvPagTdF10} a strict correspondence between the number of steps of the cut elimination of a net and its interpretation, which is still relevant here (see Theorem~\ref{theorem:qualitativeSN} and Theorem~\ref{theorem:exactSN}).

\begin{definition}[Ground-structure]\label{def:gnet}
A \emph{\grnet}, or \emph{\gnet} for short, is a finite (possibly empty) labelled directed acyclic graph whose nodes (also called links) are defined together with an arity and a coarity, i.e.\ a given number of incident edges called the \emph{premises} of the node and a given number of emergent edges called the \emph{conclusions} of the node. The valid nodes are:
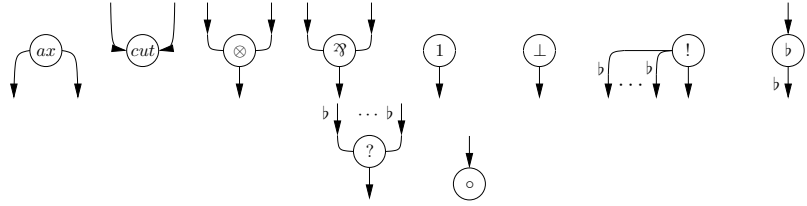
\begin{center}
\scalebox{\scalefact}{\begin{picture}(0,0)%
\includegraphics{linkax.pstex}%
\end{picture}%
\setlength{\unitlength}{3947sp}%
\begingroup\makeatletter\ifx\SetFigFontNFSS\undefined%
\gdef\SetFigFontNFSS#1#2#3#4#5{%
  \reset@font\fontsize{#1}{#2pt}%
  \fontfamily{#3}\fontseries{#4}\fontshape{#5}%
  \selectfont}%
\fi\endgroup%
\begin{picture}(684,924)(2359,-523)
\put(2706,-109){\makebox(0,0)[b]{\smash{{\SetFigFontNFSS{10}{12.0}{\rmdefault}{\mddefault}{\updefault}{\color[rgb]{0,0,0}$ax$}%
}}}}
\end{picture}%
}\quad
\scalebox{\scalefact}{\begin{picture}(0,0)%
\includegraphics{linkcut.pstex}%
\end{picture}%
\setlength{\unitlength}{3947sp}%
\begingroup\makeatletter\ifx\SetFigFontNFSS\undefined%
\gdef\SetFigFontNFSS#1#2#3#4#5{%
  \reset@font\fontsize{#1}{#2pt}%
  \fontfamily{#3}\fontseries{#4}\fontshape{#5}%
  \selectfont}%
\fi\endgroup%
\begin{picture}(624,924)(2389,-523)
\put(2706,-109){\makebox(0,0)[b]{\smash{{\SetFigFontNFSS{10}{12.0}{\rmdefault}{\mddefault}{\updefault}{\color[rgb]{0,0,0}$cut$}%
}}}}
\end{picture}%
}\quad
\scalebox{\scalefact}{\begin{picture}(0,0)%
\includegraphics{linktens.pstex}%
\end{picture}%
\setlength{\unitlength}{3947sp}%
\begingroup\makeatletter\ifx\SetFigFontNFSS\undefined%
\gdef\SetFigFontNFSS#1#2#3#4#5{%
  \reset@font\fontsize{#1}{#2pt}%
  \fontfamily{#3}\fontseries{#4}\fontshape{#5}%
  \selectfont}%
\fi\endgroup%
\begin{picture}(684,924)(2359,-523)
\put(2706,-109){\makebox(0,0)[b]{\smash{{\SetFigFontNFSS{10}{12.0}{\rmdefault}{\mddefault}{\updefault}{\color[rgb]{0,0,0}$\ltens$}%
}}}}
\end{picture}%
}\quad
\scalebox{\scalefact}{\begin{picture}(0,0)%
\includegraphics{linkpar.pstex}%
\end{picture}%
\setlength{\unitlength}{3947sp}%
\begingroup\makeatletter\ifx\SetFigFontNFSS\undefined%
\gdef\SetFigFontNFSS#1#2#3#4#5{%
  \reset@font\fontsize{#1}{#2pt}%
  \fontfamily{#3}\fontseries{#4}\fontshape{#5}%
  \selectfont}%
\fi\endgroup%
\begin{picture}(684,924)(2359,-523)
\put(2706,-109){\makebox(0,0)[b]{\smash{{\SetFigFontNFSS{10}{12.0}{\rmdefault}{\mddefault}{\updefault}{\color[rgb]{0,0,0}$\lpar$}%
}}}}
\end{picture}%
}\quad
\scalebox{\scalefact}{\begin{picture}(0,0)%
\includegraphics{linkone.pstex}%
\end{picture}%
\setlength{\unitlength}{3947sp}%
\begingroup\makeatletter\ifx\SetFigFontNFSS\undefined%
\gdef\SetFigFontNFSS#1#2#3#4#5{%
  \reset@font\fontsize{#1}{#2pt}%
  \fontfamily{#3}\fontseries{#4}\fontshape{#5}%
  \selectfont}%
\fi\endgroup%
\begin{picture}(684,924)(2359,-523)
\put(2706,-109){\makebox(0,0)[b]{\smash{{\SetFigFontNFSS{10}{12.0}{\rmdefault}{\mddefault}{\updefault}{\color[rgb]{0,0,0}$\lone$}%
}}}}
\end{picture}%
}\quad
\scalebox{\scalefact}{\begin{picture}(0,0)%
\includegraphics{linkbot.pstex}%
\end{picture}%
\setlength{\unitlength}{3947sp}%
\begingroup\makeatletter\ifx\SetFigFontNFSS\undefined%
\gdef\SetFigFontNFSS#1#2#3#4#5{%
  \reset@font\fontsize{#1}{#2pt}%
  \fontfamily{#3}\fontseries{#4}\fontshape{#5}%
  \selectfont}%
\fi\endgroup%
\begin{picture}(684,924)(2359,-523)
\put(2706,-109){\makebox(0,0)[b]{\smash{{\SetFigFontNFSS{10}{12.0}{\rmdefault}{\mddefault}{\updefault}{\color[rgb]{0,0,0}$\lbot$}%
}}}}
\end{picture}%
}\quad
\scalebox{\scalefact}{\input{linkbang.pstex_t}}\quad
\scalebox{\scalefact}{\input{linkf.pstex_t}}\quad
\scalebox{\scalefact}{\input{linkwhy.pstex_t}}\quad
\scalebox{\scalefact}{\begin{picture}(0,0)%
\includegraphics{linknew.pstex}%
\end{picture}%
\setlength{\unitlength}{3947sp}%
\begingroup\makeatletter\ifx\SetFigFont\undefined%
\gdef\SetFigFont#1#2#3#4#5{%
  \reset@font\fontsize{#1}{#2pt}%
  \fontfamily{#3}\fontseries{#4}\fontshape{#5}%
  \selectfont}%
\fi\endgroup%
\begin{picture}(684,619)(2359,-218)
\put(2706,-109){\makebox(0,0)[b]{\smash{{\SetFigFont{10}{12.0}{\rmdefault}{\mddefault}{\updefault}{\color[rgb]{0,0,0}$\circ$}%
}}}}
\end{picture}%
}
\end{center}
An edge may have or may not have a $\flat$ label: an edge with no label (resp. with a $\flat$ label) is called \emph{logical} (resp. \emph{structural}). The $\flat$-nodes have a logical premise and a structural conclusion, the $?$-nodes have $k\geq 0$ structural premises and one logical conclusion, the $!$-nodes have no premise, exactly one logical conclusion, also called \emph{main} conclusion of the node, and $k\geq0$ structural conclusions, called \emph{auxiliary} conclusions of the node. Premises and conclusions of the nodes $ax$, $cut$, $\ltens$, $\lpar$, $1$, $\bot$ are logical edges. 
Premises of the nodes $\circ$ are called \emph{conclusions} of the \gnet; we consider that a \gnet{} is given with an order $(c_1,\dots,c_n)$ of its conclusions.\\
We denote by $!(\alpha)$ the set of $!$-links of a \gnet{} $\alpha$. 
\end{definition}

When drawing a \gnet{} we order its conclusions from left to right. Also we represent edges oriented top-down so that we speak of moving upwardly or downwardly in the graph, and of nodes or edges ``above'' or ``under'' a given node/edge. In the sequel we will not write explicitly the orientation of the edges. Moreover we will not represent the $\circ$-nodes. 
In order to give more concise pictures, when not misleading, we may represent an arbitrary number of $\flat$-edges (possibly zero) as a $\flat$-edge with a diagonal stroke drawn across (see Fig~\ref{fig:notationpicture}). In the same spirit, a $?$-link with a diagonal stroke drawn across its conclusion represents an arbitrary number of $?$-links, possibly zero (see Fig~\ref{fig:notationpicture}). 

\begin{figure}
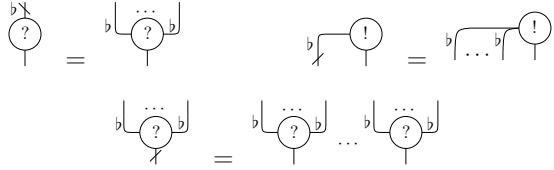

\centering
\subfigure{$$\xymatrix@C=3pt{\scalebox{\scalefact}{\input{linkwhystroke1.pstex_t}}&=&\scalebox{\scalefact}{\input{linkwhystroke2.pstex_t}}}$$}
\qquad\qquad
\subfigure{$$\xymatrix@C=3pt{\scalebox{\scalefact}{\input{linkbangstroke2.pstex_t}}&=&\scalebox{\scalefact}{\input{linkbangstroke1.pstex_t}}}$$}
\qquad\qquad
\subfigure{$$\xymatrix@C=3pt{\scalebox{\scalefact}{\input{linkwhystroke3.pstex_t}}&=&\scalebox{\scalefact}{\input{linkwhystroke4.pstex_t}}}$$}
\caption{Some conventions to picture an arbitrary number of nodes/edges}\label{fig:notationpicture}
\hrulefill
\end{figure}

\begin{definition}[Untyped $\flat$-structure, untyped nets]\label{def:struct}
For any $d \in \integers$, we define, by induction on $d$, the set of \emph{untyped $\flat$-structures} of depth $d$. 

An \emph{untyped $\flat$-structure}, or simply \emph{$\flat$-structure}, $\pi$ \emph{of depth $0$} is a \gnet{} without $!$-nodes; in this case, we set $\ground{\pi}=\pi$. An \emph{untyped $\flat$-structure} $\pi$ \emph{of depth $d+1$} is a \gnet{} $\alpha$, denoted by $\ground{\pi}$, with a function that assigns to every $!$-link $o$ of $\alpha$ with $n_o+1$ conclusions a $\flat$-structure of depth at most $d$, that we denote $\pi^o$ and we call the \emph{box of $o$}, with $n_o$ structural conclusions, also called \emph{auxiliary conclusions} of $\pi^o$, and exactly one logical conclusion, called the \emph{main conclusion} of $\pi^o$, and a bijection from the set of the $n_o$ structural conclusions of the link $o$ to the set of the $n_o$ structural conclusions of the $\flat$-structure $ \pi^o $. Moreover $\alpha$ has at least one $!$-link with a box of depth $d$.\\
We say that $\ground{\pi}$ is the \emph{\gnet{} of depth $0$ of $\pi$}; a \emph{\gnet{} of depth $d+1$ in $\pi$} is a \gnet{} of depth $d$ of the box associated by $\pi$ with a $!$-node of $\ground{\pi}$. A \emph{link $l$ of depth $d$ of $\pi$} is a link of a \gnet{} of depth $d$ of $\pi$; we denote by $\depth{l}$ the depth $d$ of $l$. We refer more generally to a link/\gnet{} of $\pi$ meaning a link/\gnet{} of some depth of $\pi$. 

A \emph{switching} of a \gnet{} $\alpha$ is an undirected subgraph of $\alpha$ obtained by forgetting the orientation of $\alpha$'s edges, by deleting one of the two premises of 
each $\lpar$-node, and for every $?$-node $l$ with $n\geq 1$ premises, by erasing all but one premises of $l$.\\
An \emph{untyped $\flat$-net}, \emph{$\flat$-net} for short, is a $\flat$-structure $\pi$ s.t. every switching of every \gnet{} of $\pi$ is an acyclic graph. An \emph{untyped net}, \emph{net} for short, is a $\flat$-net with no structural conclusion.
\end{definition}

In order to make visual the correspondence between a conclusion of a $!$-link and the associated conclusion of the box of that $!$-link, we represent the two edges by a single line crossing the border of the box (for example see Fig.~\ref{fig:cut}). 

Notice that with every structural edge $b$ of a net is associated exactly one $\flat$-node (above it) and one $\wn$-node (below it): we will refer to these nodes as \emph{the $\flat$-node/$?$-node associated with $b$}. Observe that the $\flat$-node and the $?$-node associated with a given edge might have a different depth. 

Concerning the presence of empty nets, notice that the empty net does exist and it has no conclusion. Its presence is required by the cut elimination procedure (Definition~\ref{def:Scutreduction}): the elimination of a cut between a $1$-link and a $\bot$-link yields the empty graph, and similarly for a cut between a $!$-link with no auxiliary conclusion and a $0$-ary $?$-link.
On the other hand, notice also that with a $!$-link $o$ of a net, it is \emph{never} possible to associate the empty net: $o$ has at least one conclusion and this has also to be the case for the net associated with $o$.

\begin{definition}[Size of nets]\label{def:sizeps}
The \emph{size $\sizenet{\alpha}$ of a \gnet{} $\alpha$} is the number of logical edges of $\alpha$. The \emph{size $\sizenet{\pi}$ of a $\flat$-structure $\pi$} is defined by induction on the depth of $\pi$, as follows: $\sizenet{\pi} = \sizenet{\ground{\pi}} + \sum_{o \in !(\ground{\pi})} \sizenet{\pi^o}$.
\end{definition}

\begin{wrapfigure}{l}{3.9cm}
\centering
\vspace{-5pt}
\subfigure{\scalebox{\scalefact}{\input{clash1.pstex_t}}}\quad
\subfigure{\scalebox{\scalefact}{\input{clash2.pstex_t}}}\hspace{-10pt}
\caption{Two clashes}\label{fig:clashes}
\end{wrapfigure}

Since we are in an untyped framework, nets may contain ``pathological'' cuts which are not reducible. They are called \emph{clashes} and their presence is in contrast with what happens in $\lambda$-calculus, where the simpler grammar of terms avoids clashes also in an untyped framework.

\begin{definition}[Clash]\label{def:clash}
The two edges premises of a cut-link are \emph{dual} when:
\begin{itemize}
\item they are conclusions of resp.\ a $\otimes$-node and of a $\lpar$-node, or
\item they are conclusions of resp.\ a $1$-node and of a $\bot$-node, or
\item they are conclusions of resp.\ a $!$-node and of a $?$-node.
\end{itemize}
A cut-link is a \emph{clash}, when the premises of the cut-node are not dual edges and none of the two is the conclusion of an $ax$-link.
\end{definition}

\begin{figure}
$\xymatrix@C=5pt@R=5pt{
\mathbf{(ax)}: &\scalebox{\scalefact}{\input{cutax1.pstex_t}} & \onecut & \scalebox{\scalefact}{\input{cutax2.pstex_t}}\\
\mathbf{(\ltens/\lpar)}: &\scalebox{\scalefact}{\input{cutmul1.pstex_t}} &\onecut & \scalebox{\scalefact}{\input{cutmul2.pstex_t}}\\
\mathbf{(1/\bot)}: &\scalebox{\scalefact}{\input{cut1bot.pstex_t}} &\onecut & \text{empty graph}}$
\caption{Cut elimination for nets (multiplicatives).}\label{fig:cut:multiplicatives}
\end{figure}

\begin{definition}[Cut elimination, Figures~\ref{fig:cut:multiplicatives} and~\ref{fig:cut}]\label{def:Scutreduction}
The cut elimination procedure (\cite{CarvPagTdF10}) actually comes from~\cite{hilbert}. 
To eliminate a cut $t$ in a net $\pi$ means in general to transform $\pi$ into a net\footnote{The fact that $t(\pi)$ is indeed a net should be checked, see for example~\cite{phdregnier}.} $t(\pi)$ by substituting a specific subgraph $\beta$ of $\pi$ with a graph $\beta'$ having the same pending edges (i.e. edges with no target or no source) as $\beta$. The graphs $\beta$ and $\beta'$ depend on the cut $t$ and are described in Figures~\ref{fig:cut:multiplicatives} and~\ref{fig:cut}. We also refer to $t(\pi)$ as a one step reduct of $\pi$, and to the transformations associated with the different types of cut-link as the \emph{reduction steps}. 

When one of the two premises of $t$ is a $?$-link with no premises and the other one is a $!$-link, we say that $t$ is \emph{erasing} and the reduction step is an erasing step. We write $\pi \onecut \pi'$, when $\pi'$ is the result of one reduction step and $\pi \oneerasing \pi'$ (resp.\ $\pi \onenonerasing \pi'$) in case the reduction step is (resp.\ is not) erasing. 

A cut-link $t$ of $\pi$ is \emph{stratified non-erasing}, when it is non-erasing and, for every non erasing cut (except clashes) $t'$ of $\pi$, we have $\textrm{depth}(t) \leq \textrm{depth}(t')$.  A stratified non-erasing reduction step is a step reducing a stratified non-erasing cut; we write $\pi \onestratnonerasing \pi'$ when $\pi'$ is the result of one stratified non-erasing reduction step.\\
A cut-link $t$ of $\pi$ is \emph{antistratified erasing}, when every cut-link of $\pi$ is erasing and for every cut-link $t'$ of $\pi$ we have $\textrm{depth}(t')\leq\textrm{depth}(t)$. An antistratified erasing reduction step is a step reducing an antistratified erasing cut; we write $\pi \oneantistraterasing \pi'$ when $\pi'$ is the result of one antistratified erasing reduction step.

The reflexive and transitive closure of the rewriting rules previously defined is denoted by adding a $\ast$: for example $\stratnonerasingred$ is the reflexive and transitive closure of $\onestratnonerasing$. A net $\pi$ is \emph{normalizable} if there exists a cut-free net $\pi_0$ such that $\pi \cutred \pi_0$. We denote by $\textbf{WN}$ the set of normalizable nets.\\
A \emph{reduction sequence} $R$ from $\pi$ to $\pi'$ is a sequence (possibly empty in case $\pi=\pi'$) of reduction steps $\pi \onecut \pi_1 \onecut \dots \onecut \pi_n=\pi'$. The integer $n$ is the \emph{length} of the reduction sequence. A reduction sequence $R$ is a \emph{stratified non-erasing reduction} (resp.\ an \emph{antistratified erasing reduction}) when every step of $R$ is stratified non-erasing (resp.\ antistratified erasing). A net is $\nonerasing$-normal when it contains only erasing cut-links. We denote by $\textbf{WN}^{\nonerasing}$ the set of nets $\pi$ such that there exists a $\nonerasing$-reduction sequence from $\pi$ to some $\nonerasing$-normal net.

We denote by $\textbf{SN}$ (resp.\ $\textbf{SN}^{\nonerasing}$, $\textbf{SN}^{\stratnonerasing}$) the set of nets $\pi$ such that
every reduction sequence (resp.\ $\nonerasing$-reduction sequence, $\stratnonerasing$-reduction sequence) from $\pi$ is finite and none of the reducts (resp.\ $\nonerasing$-reducts, $\stratnonerasing$-reducts) of $\pi$ contains a clash. The nets of $\textbf{SN}$ are also called \emph{strongly normalizable}.

For any net $\pi$, we set\footnote{We use here (and we will use in the sequel) K\"onig's lemma applied to countable graphs, since all the reduction relations we consider in the paper are finitely branching.} 
\begin{itemize}
\item $\nonerasingstrong{\pi} = \left\lbrace \begin{array}{ll} 
\max \{ \length{R} ; R \textrm{ is a $\nonerasing$-reduction sequence from $\pi$} \} & \textrm{if $\pi \in \textbf{SN}^{\nonerasing}$;}\\
\infty & \textrm{otherwise;} \end{array} \right.$
\item and $\strong{\pi} = \left\lbrace \begin{array}{ll} 
\max \{ \length{R} ; R \textrm{ is a reduction sequence from $\pi$} \} & \textrm{if $\pi \in \textbf{SN}$;}\\
\infty & \textrm{otherwise.} \end{array} \right.$
\end{itemize}
\end{definition}

\begin{figure}
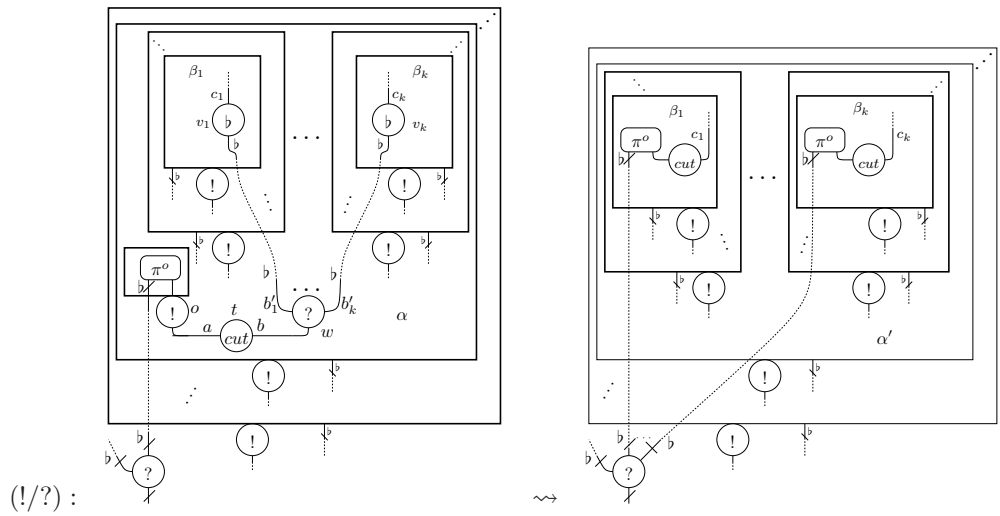

$\xymatrix@C=5pt@R=5pt{
\mathbf{(!/?)}: &\scalebox{\scalefact}{\input{cutbangwhy1.pstex_t}} & \onecut & \scalebox{\scalefact}{\input{cutbangwhy3.pstex_t}} 
}$
\caption{Cut elimination for nets. In the $(!/?)$ case what happens is that the $!$-link $o$ dispatches $k$ copies of $\pi^o$ ($k\geq 0$ being the arity of the $?$-node $w$ premise of the cut) inside the $!$-boxes (if any) containing the $\flat$-nodes associated with the premises of $w$; notice also that the reduction duplicates $k$ times the premises of $?$-nodes which are associated with the auxiliary conclusions of $o$.}\label{fig:cut}
\hrulefill
\end{figure}

\begin{remark}\label{rem:ClashesSemantics}
Notice that the presence of clashes induces a slight difference between the definition of  ``normalizable net'' and that of  ``strongly normalizable net'': a normalizable net $\pi$ (so as its reducts) might contain a clash, which is not the case of a strongly normalizable net (nor of its reducts). This is consistent with the basic intuition behind these two notions: from a normalizable net one should be able (by means of ``correct'' computations) to reach a normal form, while from a strongly normalizable net one should be able by reducing at any time any cut to reach a normal form, so that such nets can never contain clashes.

In a pure rewriting approach, one could consider a different notion of weakly and strongly normalizable net: in~\cite{SNLL10} normal nets can contain clashes (see Subsection 2.4 p.420 of~\cite{SNLL10}). This cannot be accepted here (and was already excluded in~\cite{CarvPagTdF10} for the same reasons), since a clash in a net immediately yields an empty interpretation of the net (see the next Section~\ref{sect:experiments}), from which no information can be extracted, and certainly not the number of steps leading to a normal form.
\end{remark}

\begin{definition}[Ancestor, residue]\label{def:ancestors}
Let $\pi \onecut \pi'$. When an edge $d$ (resp. a node $l$) of $\pi'$ comes from a (unique) edge $\overleftarrow{d}$ (resp. node $\overleftarrow{l}$) of $\pi$, we say that $\overleftarrow{d}$ (resp. $\overleftarrow{l}$) is the \emph{ancestor} of $d$ (resp. $l$) in $\pi$ and that $d$ (resp. $l$) is a \emph{residue} of $\overleftarrow{d}$ (resp. $\overleftarrow{l}$) in $\pi'$. If this is not the case, then $d$ (resp. $l$) has no ancestor in $\pi$, and we say it is a \emph{created} edge (resp. node). We indicate, for every type of cut elimination step of Fig.~\ref{fig:cut}, which edges (resp. links) are created in $\pi'$ (meaning that the other edges/nodes of $\pi'$ are residues of some $\pi$'s edge/node). We use the notations of Figures~\ref{fig:cut:multiplicatives} and~\ref{fig:cut}:
\begin{minilist}
\item
$(ax)$: 
there are no created edges, nor created nodes in $\pi'$. Remark that $a,b$ are erased in $\pi'$, so that we consider $c$ in $\pi'$ as the residue of $c$ in $\pi$;
\item
$(\otimes/\lpar)$: 
there are no created edges, while the two new cut-links between the two left (resp.\ right) premises of the $\lpar$- and $\otimes$-links are created nodes;
\item
$(\lone/\bot)$: 
there are no created edges, nor created nodes in $\pi'$;
\item
$(!/?)$:
every auxiliary conclusion added to the $!$-links containing one copy of $\pi^{o}$ is a created edge; every cut link between (a copy of) $\pi^{o}$'s main conclusion and $c_i$ is a created node.\footnote{Notice that every $!$-link of $\pi'$ which contains a copy of $\pi^o$ is considered a residue of the corresponding $!$-link of $\pi$, even though it has different auxiliary conclusions. Notice also that the edges/nodes in each copy of $\pi^{o}$ are considered residues of the corresponding edges/nodes in $\pi^o$.}
\end{minilist}
\end{definition}

\subsection{The non-erasing normalization and the stratified normalization}\label{subsect:netsReduction}

In order to prove our main qualitative result (Theorem~\ref{theorem:qualitativeSN}), we reduce strong normalization to non erasing strong normalization: this is Proposition~\ref{proposition:SN=SNnonerasing}. We actually prove a variant of a very similar result proven in~\cite{SNLL10}: the difference is related to the way one handles clashes (Remark~\ref{rem:ClashesSemantics}).\\
In order to measure by semantic means the exact length of the longest reduction sequence(s) starting from a given strongly normalizable net (Theorem~\ref{theorem:exactSN}), we show that there always exists such a sequence consisting first of non erasing stratified steps and then of erasing antistratified steps: this is Proposition~\ref{proposition:nonerasing-erasing}. 

\bigskip

The first step is rather standard in spirit: one proves that erasing steps can always be ``postponed'' (Proposition~\ref{proposition:postponingerasing}).

\begin{lemma}\label{lemma:postponing:erasing}
Assume that $\pi \oneerasing \pi_1$ and $\pi_1 \onenonerasing \pi''$. Then there exist $\pi'$ such that $\pi \onenonerasing \pi'$ and a reduction sequence $\pi' \cutred \pi''$:
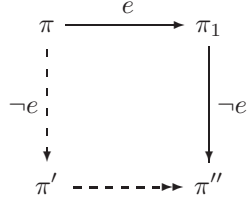
\begin{diagram}
\pi & \rTo^{\erasing} & \pi_1 \\
\dDashTo^{\nonerasing} & & \dTo_{\nonerasing}^{} \\
\pi' & \rDashtrTo_{} & \pi''
\end{diagram}
\end{lemma}

\begin{proof}
See Lemma 4.4 p. 431 of~\cite{SNLL10}.
\end{proof}

\begin{proposition}[postponing erasing steps]\label{proposition:postponingerasing}
For any net $\pi_0$ such that there is no infinite reduction sequence from $\pi_0$, for any finite reduction sequence $R$ from $\pi_0$ to $\pi'$, there exist a $\nonerasing$-reduction sequence $R'$ from $\pi_0$ to some net $\pi$ and an $\erasing$-reduction sequence $R_0$ from $\pi$ to $\pi'$ such that $\length{R} \leq \length{R'} + \length{R_0}$.
\end{proposition}

\begin{proof}
By induction on $\max \{ \length{R} ; R \textrm{ is a reduction sequence from $\pi_0$} \}$. Let $R$ be a finite reduction sequence $\pi_0 \onecut \pi_1 \onecut \ldots \pi_{n-1} \onecut \pi_n = \pi'$. If $R$ has no $\nonerasing$-reduction steps, then we set $\pi = \pi_0$ and $R_0 = R$. Otherwise, we set $k = \min \{ i \in \integers ; \pi_i \oneerasing \pi_{i+1} \}$: if $k > 0$, then we apply the induction hypothesis to $\pi_1$; if $k = 0$, then we set $r = \min \{ j \in \integers ; \pi_{j} \onenonerasing \pi_{j+1} \}$; we apply $r$ times Lemma~\ref{lemma:postponing:erasing}, we thus obtain a reduction sequence $R_1$ from $\pi_0$ to $\pi_{r+1}$ 
in which the first reduction step $\pi_0 \onecut \pi'_1$ is non-erasing. We can thus consider the reduction sequence $R_1$ followed by the reduction sequence $\pi_{r+1} \onecut \pi_{r+2} \onecut \ldots \pi_{n-1} \onecut \pi_n$ and apply the induction hypothesis to $\pi'_1$.
\end{proof}

To prove $\textbf{SN} = \textbf{SN}^{\nonerasing}$, we apply the techniques of~\cite{SNLL10}, taking care of clashes (Fact~\ref{fact:erasing steps do not produce clashs}).

\begin{fact}\label{fact:erasing steps do not produce clashs}
If $\pi \erasingred \pi'$ and $\pi'$ contains some clash, then the net $\pi$ contains some clash too.
\end{fact}

\begin{proof}
If $\pi \oneerasing \pi'$, then every edge of $\pi'$ has an ancestor in $\pi$. Now, the ancestor of a clash is always a clash too.
\end{proof}

\begin{proposition}\label{proposition:SN=SNnonerasing}
We have $\textbf{SN} = \textbf{SN}^{\nonerasing}$.
\end{proposition}


\begin{proof}
If $\pi \notin \textbf{SN}$, then we are in one the two following cases:
\begin{enumerate}
\item 
\begin{itemize}
\item there is no infinite reduction sequence from $\pi$
\item and there is some net $\pi'$ with some clash such that $\pi \cutred \pi'$,
\end{itemize}
\item or there exists an infinite reduction sequence from $\pi$.
\end{enumerate}

Assume that we are in the first case. Then, by Proposition~\ref{proposition:postponingerasing}, there exist a $\nonerasing$-reduction sequence $R$ from $\pi$ to $\pi_1$ and an $\erasing$-reduction sequence from some net $\pi_1$ to $\pi'$. Since $\pi'$ is a net containing some clash, by Fact~\ref{fact:erasing steps do not produce clashs}, the net $\pi_1$ contains some clash too, hence $\pi \notin \textbf{SN}^{\nonerasing}$.

Now, if we are in the second case, one can show that there exists an infinite $\nonerasing$-reduction sequence from $\pi$. This has been proven in~\cite{SNLL10} using Lemma~\ref{lemma:postponing:erasing}: see Proposition 4.5 p. 431 of~\cite{SNLL10}.
\end{proof}

We now turn to the proof of Proposition~\ref{proposition:nonerasing-erasing}, which essentially consists, given a strongly normalizing net $\pi$, in turning any reduction sequence starting from $\pi$ into a ``canonical'' reduction sequence: a $\stratnonerasing$-reduction sequence followed by an antistratified erasing reduction sequence. We show that this transformation never shortens the length of reduction sequences, which entails that among the longest reduction sequences starting from $\pi$, there always exists a canonical one. 
The first step is to prove that one can always reach a $\nonerasing$-normal net by means of a $\stratnonerasing$-reduction sequence of maximum length (Proposition~\ref{proposition:stratisworse}), the second step is to relate the number of cut-links of a (strongly normalizable) net to the length of canonical reduction sequences (Lemma~\ref{lemma:nonerasing-erasing}).

\begin{lemma}\label{lemma:confluence:nonerasing-stratnonerasing}
Assume that $\pi \onenonerasing \pi_1$ and $\pi \onestratnonerasing \pi'$ with $\pi' \not= \pi_1$. Then there exist $\pi''$ such that $\pi_1 \onestratnonerasing \pi''$ and a non-empty reduction sequence $\pi' \nonerasingred \pi''$:
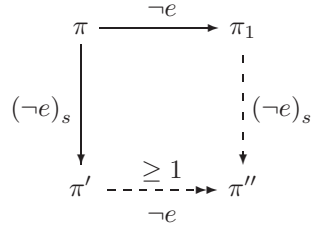
\begin{diagram}
\pi & \rTo^{\nonerasing} & \pi_1 \\
\dTo^{\stratnonerasing} & & \dDashTo_{\stratnonerasing} \\
\pi' & \rDashtrTo_{\nonerasing}^{\geq 1} & \pi''
\end{diagram}
\end{lemma}

\begin{proof}
Let $x$ (resp.\ $y$) be the cut-link reduced by the step $\pi \onenonerasing \pi_1$ (resp.\ $\pi \onestratnonerasing \pi'$): we know by hypothesis that $x\neq y$. Since $x$ is non erasing and $y$ is stratified, there exists a unique residue $y^{1}$ of $y$ in $\pi_{1}$. Since $y$ is non erasing and $x$ needs not being stratified, there exist $n\geq 1$ residues $x'_{1},\ldots,x'_{n}$ of $x$ in $\pi'$. The net $\pi''$ can be obtained both by reducing $y^{1}$ in $\pi_{1}$ and by reducing $x'_{1},\ldots,x'_{n}$ in $\pi'$.
\end{proof}

In the sequel, we use the (obvious) fact that whenever there exists a non erasing cut-link in a net, there also exists a \emph{stratified} non erasing cut-link in that same net.

\begin{proposition}\label{proposition:stratisworse}
For any $\pi_0 \in \textbf{SN}^{\nonerasing}$, for any $\nonerasing$-reduction sequence $R'''$ from $\pi_0$ to a $\nonerasing$-normal form $\pi$, there exists a $\stratnonerasing$-reduction sequence $R_1$ from $\pi_0$ to $\pi$ such that $\length{R'''} \leq \length{R_1}$.
\end{proposition}

\begin{proof}
We prove, by induction on $\nonerasingstrong{\pi_0}$, that, for any $\pi_0 \in \textbf{SN}^{\nonerasing}$, for any $\nonerasing$-reduction sequence $R'''$ from $\pi_0$ to a $\nonerasing$-normal form $\pi$, for any $\pi'$ such that $\pi \onestratnonerasing \pi'$, there exists a $\stratnonerasing$-reduction sequence $R_1$ from $\pi'$ to $\pi$ such that $\length{R'''} \leq \length{R_1} +1$. 
\begin{itemize}
\item If $\textit{strong}_{\nonerasing}(\pi_0) = 0$, then there is no such $\pi'$.
\item If $\textit{strong}_{\nonerasing}(\pi_0) > 0$, then we apply Lemma~\ref{lemma:confluence:nonerasing-stratnonerasing} and the induction hypothesis. More precisely, suppose that $R'''$ is such that $\pi_{0} \onenonerasing \pi_{1} \nonerasingred \pi$. If $\pi' = \pi_{1}$, then we apply the induction hypothesis to $\pi_{1}$. Otherwise, $\pi' \neq \pi_{1}$ and $\pi_{0} \onestratnonerasing \pi'$, so we can apply Lemma~\ref{lemma:confluence:nonerasing-stratnonerasing}: there exist $\pi''$ such that $\pi_1 \onestratnonerasing \pi''$ and a non-empty reduction sequence $\pi' \nonerasingred \pi''$. We can call $R'''_{1}$ the $\nonerasing$-reduction sequence leading from $\pi_{1}$ to $\pi$ and apply the induction hypothesis to $\pi_{1}$: there exists a $\nonerasing$-reduction sequence $R_{1}^1$ from $\pi''$ to $\pi$ such that $\length{R'''_1} \leq\length{R_1^1}+1$. Now, since there exists a non-empty reduction sequence $\pi' \nonerasingred \pi''$, there also exists a $\nonerasing$-reduction sequence $R'''_{2}$ from $\pi'$ to $\pi$ such that $\length{R'''_{2}} \geq \length{R_{1}^1} + 1$. By applying the induction hypothesis to $\pi'$, there exists a $\stratnonerasing$-reduction sequence $R'''_3$ from $\pi'$ to $\pi$ such that $\length{R'''_3} \geq \length{R'''_2}$. We consider $R_1$ defined by $\pi_0 \onestratnonerasing \pi'$ followed by $R'''_3$. We have $\length{R_1} = \length{R'''_3} +1 \geq \length{R'''_2} +1 \geq \length{R_1^1} + 1 +1 \geq \length{R'''_1} +1 = \length{R'''}$.
\end{itemize}
\end{proof}

\begin{fact}\label{fact:nonerasing reduction do not erase cut-links}
If $\pi \onenonerasing \pi'$, then $\pi'$ has at least $n-1$ cut-links, where $n$ is the number of cut-links in $\pi$.
\end{fact}

\begin{proof}
If $\pi' = t(\pi)$ with $t$ a non-erasing cut-link, then every cut-link of $\pi$, except $t$, has at least one residue in $\pi'$.
\end{proof}

\begin{lemma}\label{lemma:nonerasing-erasing}
Let $\pi_0 \in \textbf{SN}$ with at least $n$ cut-links. Then there exist
\begin{itemize}
\item a $\nonerasing$-normal net $\pi$;
\item a $\nonerasing$-reduction sequence $R_1$ from $\pi_0$ to $\pi$;
\item and an antistratified $\erasing$-reduction sequence $R_2$ from $\pi$
\end{itemize}
such that $n \leq \length{R_1} + \length{R_2}$.
\end{lemma}

\begin{proof}
By induction on $\strong{\pi_0}$. We distinguish between two cases:
\begin{itemize}
\item There exists $\pi_1$ such that $\pi_0 \onenonerasing \pi_1$: we apply Fact~\ref{fact:nonerasing reduction do not erase cut-links} and the induction hypothesis on $\pi_1$.
\item The net $\pi_0$ is $\nonerasing$-normal: we take for $R_1$ the empty reduction sequence from $\pi_0$ to $\pi_0$ and for $R_2$ an antistratified $\erasing$-reduction sequence $\pi_0 \oneerasing \pi_1 \ldots \oneerasing \pi_n$ such that, for any $i \in \{ 0, \ldots, n \}$, the net $\pi_i$ has exactly $k-i$ erasing cut-links, where $k$ is the number of cut-links of $\pi_0$.
\end{itemize}
\end{proof}

\begin{fact}\label{fact:number of cut-links}
Let $R_0$ be an $\erasing$-reduction sequence from $\pi'$. Then $\pi'$ has at least $\length{R_0}$ cut-links.
\end{fact}

\begin{proof}
If $\pi \oneerasing \pi' = t(\pi)$, then
\begin{itemize}
\item every cut-link of $\pi'$ has an ancestor in $\pi$
\item and $t$ has no residue in $\pi'$;
\end{itemize}
hence the number of cut-links in $\pi'$ is strictly smaller than the number of cut-links in $\pi$.
\end{proof}

\begin{proposition}\label{proposition:nonerasing-erasing}
For any $\pi_0 \in \textbf{SN}$, there exist a $\stratnonerasing$-reduction sequence $R_1 : \pi_0 \stratnonerasingred \pi$ with $\pi$ $\nonerasing$-normal and an antistratified $\erasing$-reduction sequence $R_2$ from $\pi$ such that $\strong{\pi_0} = \length{R_1} + \length{R_2}$.
\end{proposition}

\begin{proof}
Let $\pi_0 \in \textbf{SN}$ and let $R$ be a reduction sequence from $\pi_0$. 
By Proposition~\ref{proposition:postponingerasing}, there exist a $\nonerasing$-reduction sequence $R'$ from $\pi_0$ to some net $\pi'$ and an $\erasing$-reduction sequence $R_0$ from $\pi'$ such that $\length{R}\leq\length{R'} + \length{R_0}$. By Fact~\ref{fact:number of cut-links}, the net $\pi'$ has at least $\length{R_0}$ cut-links, hence, by Lemma~\ref{lemma:nonerasing-erasing}, there exist
\begin{itemize}
\item a $\nonerasing$-normal net $\pi$;
\item a $\nonerasing$-reduction sequence $R''$ from $\pi'$ to $\pi$;
\item and an antistratified $\erasing$-reduction sequence $R_2$ from $\pi$
\end{itemize}
such that $\length{R_0} \leq \length{R''} + \length{R_2}$. We consider $R'''$ defined by $R'$ followed by $R''$. By Proposition~\ref{proposition:stratisworse}, there exists a $\stratnonerasing$-reduction sequence $R_{1}$ from $\pi_{0}$ to $\pi$ such that $\length{R_1}\geq\length{R'''}$. We thus have: $\length{R_1}+\length{R_2}\geq\length{R'}+\length{R''}+\length{R_{2}}\geq\length{R'}+\length{R_{0}}\geq\length{R}$. By taking as $R$ any reduction sequence such that $\length{R}=\strong{\pi_0}$, we obtain the required $R_{1}$ and $R_{2}$.
\end{proof}

When $ \pi $ (resp.\ $ \pi' $) is a net having $c$ (resp.\ $c'$) among its conclusions, we denote in the sequel by $\cutnets{\pi}{\pi'}{c}{c'}$ the net obtained by connecting $ \pi $ and $ \pi' $ by means of a $cut$-link with premises $ c $ and $ c' $.

\begin{corollary}\label{corollary:postponingerasing}
Let $ \pi $ (resp.\ $ \pi' $) be a net having $c$ (resp.\ $c'$) among its conclusions, and assume that $\cutnets{\pi}{\pi'}{c}{c'}$ is strongly normalizable. 

There exists $R_1 : \cutnets{\pi}{\pi'}{c}{c'} \stratnonerasingred \pi_1$ and $R_2 : \pi_1 \erasingred \pi_2$ antistratified such that
\begin{itemize}
\item $\pi_1$ is $\nonerasing$-normal;
\item $\pi_2$ is cut-free;
\item $\strong{\cutnets{\pi}{\pi'}{c}{c'}} = \length{R_1} + \length{R_2}$.
\end{itemize}
\end{corollary}

\section{Experiments and the interpretations of nets}\label{sect:experiments}

We introduce experiments for nets (a well-known notion coming from~\cite{ll}), adapted to our framework (Definition~\ref{def:experiment}).

In~\cite{CarvPagTdF10, MR2926280} experiments are defined in an untyped framework; we follow here the same approach in our Definition~\ref{def:experiment}. Experiments allow to compute the semantics of nets:  the \emph{interpretation} $\sm{\pi}$ of a net $\pi$ is the set of the results of $\pi$'s experiments (Definition~\ref{def:experiment}). Like in~\cite{CarvPagTdF10, MR2926280}, in the following definition the set $\{+, -\}$ is used in order to ``semantically distinguish'' cells of type $\ltens$ from cells of type $\parr$, which is mandatory in an untyped framework. The function $(\ )^{\perp}$ (which is the semantic version of linear negation) flips polarities (see Definition~\ref{def:Dfunction}).

We also introduce here another ``ad hoc interpretation'' of $\pi$, denoted by $\smbis{\pi}$, which (like $\sm{\pi}$) is a set of points that can be computed starting from $\pi$ (Definition~\ref{def:experiment}). Intuitively, every element of $\smbis{\pi}$ keeps trace of all the ``weakenings'' (the $?$-links with no premise) of $\pi$, which is not the case of all the elements of $\sm{\pi}$ (see Remark~\ref{rem:exp+expbis} for a more technical comparison): this difference will be essential in the next sections. A crucial property of $\smbis{\pi}$ is the invariance under \emph{non erasing} cut elimination (Proposition~\ref{prop:invariancesmbis}). 

\bigskip

\begin{definition}\label{def:D}
We define $D_n$ by induction on $n$:
\begin{minilist}
\item $D_0 := \{+,-\} \times (A \cup \{\ast\})$
\item $D_{n+1} := D_0 \cup (\{+,-\} \times D_n \times D_n) \cup (\{ +,- \} \times \finitemultisets{D_n})$
\end{minilist}

\bigskip

We set $D \mathrel{:=} \bigcup_{n\in\textrm{N}} D_n$, and we call \emph{rank of an element $x \in D$} (and we denote by $\rank{x}$) the least $n$ such that $x \in D_n$.

When $(+,\multiset{})$ does not appear in $x\in D$, we say that $x$ is \emph{exhaustive}\footnote{We mean here that the ordered sequence of characters $(+, [])$ is not a subsequence of $x$ (as a word).}. We denote by $\exhaustive{X}$ the set of the exhaustive elements of any given subset $X$ of $D$. When $X\subseteq D^{n}$, we denote by $\exhaustive{X}$ the set $\{(x_{1},\ldots,x_{n})\in X: x_{i}\textrm{ is exhaustive for every }i\in\{1,\ldots,n\}\}$.
\end{definition}

\begin{definition}\label{def:Dfunction}
Let $+^\perp = -$ and $-^\perp = +$. We define $x^\perp$ for any $x \in D$, by induction on $\rank{x}$:
\begin{itemize}
\item for $a \in A\cup\{\ast\}$, $(p, a)^\perp = (p^\perp, a)$;
\item for $a \in \{\ast\}$, $(p, a)^\perp = (p^\perp, a)$;
\item else, $(p, x, y)^\perp = (p^\perp, x^\perp, y^\perp)$, and $(p, [x_1,\dots,x_n])^\perp = (p^\perp, [x_1^\perp,\dots, x_n^\perp])$.
\end{itemize}
\end{definition}

A key feature is that, for every $x\in D$, one has $x\neq x^\perp$, a property already used in the proof of the main qualitative result of~\cite{CarvPagTdF10} (here Theorem~\ref{th:invariance}).

Now, we show how to compute the interpretation of an untyped net directly, without passing through a sequent calculus. This is done by adapting the notion of experiment to our untyped framework. For a net $\pi$ with $n$ conclusions, we define the \emph{$\sm{}$-interpretation of $\pi$}, denoted by $\sm{\pi}$, as a subset of $D^n$, that can be seen as a morphism of the category \textbf{Rel} of sets and relations from the interpretation of $1$ to the interpretation of $\bigparr_{i=1}^n D$. We compute $\sm{\pi}$ by means of the \emph{$\sm{}$-experiments of $\pi$}, a notion introduced by Girard in \cite{ll} and central in this paper. We introduce also a variant of this notion, the \emph{$\smbis{}$-experiments of $\pi$} that allow to compute $\smbis{\pi}$. We define, by induction on the depth of $\pi$, what the $\sm{}$-experiments and $\smbis{}$-experiments of $\pi$ are:
\begin{figure}
\centering
\scalebox{\scalefact}{\input{expax.pstex_t}}\quad\,
\scalebox{\scalefact}{\input{expcut.pstex_t}}\quad\,
\scalebox{\scalefact}{\input{expone.pstex_t}}\quad\,
\scalebox{\scalefact}{\input{expbot.pstex_t}}\quad\,
\scalebox{\scalefact}{\input{exptens.pstex_t}}\quad\,
\scalebox{\scalefact}{\input{exppar.pstex_t}}\quad\,
\scalebox{\scalefact}{\input{expf.pstex_t}}\quad\,
\scalebox{\scalefact}{\input{expwhy.pstex_t}}\quad\,
\scalebox{\scalefact}{\input{expbang.pstex_t}}
\caption{$\sm{}$-experiments of $\flat$-nets, with $x$, $y$, $x_i$ $\in D$ and $\mu_i \in \finitemultisets{D}$.}\label{fig:experiment}
\hrulefill
\end{figure}

\begin{definition}[Experiment]\label{def:experiment}
An \emph{$\sm{}$-experiment $e$ of a $\flat$-net $\pi$}, denoted by $e :_{\sm{}} \pi$, is a function which associates with every $!$-link $o$ of $\ground{\pi}$ a multiset $\multiset{e^o_1,..., e^o_k}$ of $k \geq 0$ $\sm{}$-experiments of $\pi^o$, and with every edge $a$ of $\ground{\pi}$ an element of $D$. 

An \emph{$\smbis{}$-experiment $e$ of a $\flat$-net $\pi$}, denoted by $e :_{\smbis{}} \pi$, is a function which associates with every $!$-link $o$ of $\ground{\pi}$ a multiset $\multiset{e^o_1,..., e^o_k}$ of $k > 0$ $\smbis{}$-experiments of $\pi^o$, and with every edge $a$ of $\ground{\pi}$ an element of $D$.

In the cases of $ax$-links, cut-links, $\lone$-links, $\lbot$-link, $\ltens$-links, $\lpar$-links, $\flat$-links, $!$-links and $?$-links with $n \geq 1$ premises, the standard conditions of Figure~\ref{fig:experiment} hold both for $\sm{}$-experiments and $\smbis{}$-experiments; more precisely, if $a,b,c$ are edges of $\ground{\pi}$ the following conditions hold:
\begin{itemize}
\item if $a, b$ are the conclusions (resp. the premises) of an $ax$-link (resp. cut-link), then $e(a)=e(b)^\perp$;
\item if $c$ is the conclusion of a $\lone$-link (resp. $\lbot$-link), then $e(c)=\sequence{+, *}$ (resp. $e(c)=\sequence{-, *}$);
\item if $c$ is the conclusion of a $\ltens$-link (resp. $\lpar$-link) with premises $a, b$, then $e(c)=\sequence{+,e(a),e(b)}$ (resp. $e(c)=\sequence{-,e(a),e(b)}$);
\item if $c$ is the conclusion of a $\flat$-link with premise $a$, then $e(c) =\sequence{-,[e(a)]}$; 
\item if $c$ is the conclusion of a $?$-link with premises $a_1, \dots, a_n$ where $n\geq 1$, and for every $i\leq n$, $e(a_i)=\sequence{-,\mu_i}$, where $\mu_i$ is a finite multiset of elements of $D$, then $e(c) = \sequence{-, \sum_{i \leq n} \mu_i}$;
\item if $c$ is a conclusion of a $!$-link $o$ of $\ground{\pi}$, let $\pi^o$ be the box of $o$ and $e(o)=[e^o_1, \dots, e^o_n]$. If $c$ is the logical conclusion of $o$, let $c^o$ be the logical conclusion of $\pi^o$, then $e(c) = \sequence{+,[e^o_1(c^o), \dots ,e^o_n(c^o)]}$, if $c$ is a structural conclusion of $o$, let $c^o$ be the structural conclusion of $\pi^o$ associated with $c$, and for every $i\leq n$, let 
$ e^o_i(c^o) = (-, \mu_i) $, then $e(c) = \sequence{-, \sum_{i \leq n} \mu_i}$.
\end{itemize}

In the case of a $?$-link with no premise and the edge $c$ as conclusion, we require that:
\begin{minilist}
\item
$e(c) = (-, \multiset{})$, for an $\sm{}$-experiment $e$
\item
$e(c) = (-, a)$ with $a \in \finitemultisets{D}$ 
for an $\smbis{}$-experiment $e$.
\end{minilist}

\bigskip

When $e$ is an $\sm{}$-experiment (resp. an $\smbis{}$-experiment), we set\footnote{Notice that when $e$ is an $\sm{}$-experiment one always has $\weakeningsofexperiment{e} = \multiset{}$.}:
\begin{eqnarray*}
\weakeningsofexperiment{e} & = & \sum_{\begin{array}{c} c \textrm{ is the conclusion of a $?$-link of $\ground{\pi}$ with no premise} \\ e(c) = (-, \mu) \end{array}} \mu \\
& & + \sum_{o \textrm{ is a $!$-link of $\ground{\pi}$}} \sum_{e^o \in e(o)} \weakeningsofexperiment{e^o} \enspace .
\end{eqnarray*}

If $c_1, \dots, c_n$ are the conclusions of $\pi$, then the \emph{result of $e$}, denoted by $|e|$, is the element\footnote{Recall that a \gnet, hence a $\flat$-net, is given together with an order on its conclusions, so the sequence $\sequence{e(c_1), \ldots, e(c_n)}$ is uniquely determined by $e$ and $\pi$.} $\sequence{e(c_1), \ldots, e(c_n)}$ of $D^n$. The \emph{$\sm{}$-interpretation of $\pi$} is the set of the results of its $\sm{}$-experiments. The \emph{$\smbis{}$-interpretation of $\pi$} is the set of the pairs $(\result{e}, \weakeningsofexperiment{e})$ such that $e$ is an $\experimentbis{}{\pi}$.
\begin{eqnarray*}
\sm{\pi} & \mathrel{:=} & \left\lbrace \sequence{e(c_1), \ldots, e(c_n)} \: ; \: e \text{ is an $\sm{}$-experiment of } \pi \right\rbrace \enspace ; \\
\smbis{\pi} & \mathrel{:=} & \left\lbrace (\sequence{e(c_1), \ldots, e(c_n)}, \weakeningsofexperiment{e}) \: ; \: e \text{ is an $\smbis{}$-experiment of } \pi \right\rbrace \enspace .
\end{eqnarray*}
If $\mathbf{y} = \sequence{e(c_1), \ldots, e(c_n)}$ is the result of an $\sm{}$-experiment (resp. an $\smbis{}$-experiment) $e$ of $\pi$, we denote by $\mathbf{y}_{c_i}$ the element $e(c_i)$, for every $i\leq n$. Generally, if $\mathbf{d} = \sequence{c_{i_1},\dots,c_{i_k}}$ is a sequence of conclusions of $\pi$, we note by $\mathbf{y}_{\mathbf{d}}$ the element $\sequence{e(c_{i_1}),\dots,e(c_{i_k})}$ of $\mathbf{D}$.
\end{definition}

\begin{remark}\label{rem:exp+expbis}
The difference between $\sm{}$-experiments and $\smbis{}$-experiments appears clearly in the case of a $?$-link with no premise of Definition~\ref{def:experiment}, but there is another (slightly subtler) point where it shows up: while an $\sm{}$-experiment can associate with a $!$-link of $\ground{\pi}$ an empty multiset of experiments, this cannot be the case for an $\smbis{}$-experiment. Such a (heavy) constraint forbids to ``hide'' pieces of proofs, which is mandatory if one wants to be able to speak of strong normalization.
\end{remark}

\begin{remark}\label{rem:IntersectTpesRex}
When we just consider the nets encoding $\lambda$-terms, these two different interpretations $\sm{}$ and $\smbis{}$ correspond respectively to the two following non-idempotent intersection types systems: System~R of \cite{phddecarvalho} and \cite{Carvalhoexecution} (and called System~$\mathcal{M}$ in \cite{Inhabitation}) and System~$R^\textit{ex}$:\\
\begin{itemize}
\item The set of types is defined by the following grammar:\\
$\alpha ::= \gamma \: \vert \: a \to \alpha$ (types)\\
$a ::= [\alpha_1, \ldots, \alpha_n]$ (finite multiset of types)\\
where $\gamma$ ranges over a countable set $A$ and $n \in \integers$.
\item Environments are functions from variables to finite multisets of types, assigning the empty multiset to almost all the variables. If $\Gamma_1$, \ldots, $\Gamma_m$ are $m$ environments, then we denote by $\Gamma_1, \ldots, \Gamma_m$ the environment $\Gamma$ defined by $\Gamma(x) = \sum_{i=1}^m \Gamma_i(x)$ for any variable $x$. Moreover we denote by $x : a$ the environment $\Gamma$ defined by $\Gamma(y) = \left\lbrace \begin{array}{ll} \textit{$a$} & \textit{if $y = x$;}\\ \textit{$[]$} & \textit{otherwise.} \end{array} \right.$
\item A typing judgement is a triple of the form $\Gamma \vdash_R t : \alpha$ (respectively $\Gamma \vdash_{R^\textit{ex}} t : \alpha$). The types systems are those given respectively in Figure~\ref{figure: System R} and in Figure~\ref{figure: System R^ex}.
\end{itemize}
System $R^\textit{ex}$, like the non-idempotent intersection types system considered in \cite{DBLP:conf/fossacs/BernadetL11}, \cite{bernadetleng11b} and \cite{bernadetlengrand13}, characterizes strongly normalizing $\lambda$-terms. There are some differences between the two systems. In particular, if we identify the empty multiset with the type $\omega$, then in System~$R^\textit{ex}$ then the type $\omega$ can be used for weakenings but not as a universal type.
\end{remark}

\begin{figure}
\centering
\AxiomC{}
\UnaryInfC{$x : [\alpha] \vdash_R x : \alpha$}
\DisplayProof
\begin{center} \end{center}
\AxiomC{$\Gamma, x : a \vdash_R t : \alpha$}
\RightLabel{$\Gamma(x) = []$}
\UnaryInfC{$\Gamma \vdash_R \lambda x. t : a \to \alpha$}
\DisplayProof
\begin{center} \end{center}
\AxiomC{$\Gamma_0 \vdash_R v : [\alpha_1, \ldots, \alpha_n] \to \alpha$}
\AxiomC{$(\Gamma_i \vdash_R u : \alpha_i)_{i \in \{ 1, \ldots, n\}}$}
\RightLabel{$n \in \integers$}
\BinaryInfC{$\Gamma_0, \ldots, \Gamma_n \vdash_R (v)u : \alpha$}
\DisplayProof
\caption{The type assignment system $R$ for the $\lambda$-calculus}\label{figure: System R}
\hrulefill
\end{figure}

\begin{figure}
\centering
\AxiomC{}
\RightLabel{$m \in \integers$}
\UnaryInfC{$x : [\alpha], y_1 : a_1, \ldots, y_m : a_m \vdash_{R^\textit{ex}} x : \alpha$} 
\DisplayProof
\begin{center} \end{center}
\AxiomC{$\Gamma, x : a \vdash_{R^\textit{ex}} t : \alpha$}
\RightLabel{$\Gamma(x) = []$}
\UnaryInfC{$\Gamma \vdash_{R^\textit{ex}} \lambda x. t : a \to \alpha$}
\DisplayProof
\begin{center} \end{center}
\AxiomC{$\Gamma_0 \vdash_{R^\textit{ex}} v : [\alpha_1, \ldots, \alpha_n] \to \alpha$}
\AxiomC{$(\Gamma_i \vdash_{R^\textit{ex}} u : \alpha_i)_{i \in \{ 1, \ldots, n\}}$}
\RightLabel{$n \in \integers \setminus \{ 0 \}$}
\BinaryInfC{$\Gamma_0, \ldots, \Gamma_n \vdash_{R^\textit{ex}} (v)u : \alpha$}
\DisplayProof
\caption{The type assignment system $R^\textit{ex}$ for the $\lambda$-calculus}\label{figure: System R^ex}
\hrulefill
\end{figure}

In case the net $\pi$ is cut-free, $\smbis{}$-experiments ``can choose'' the finite multiset $a$ such that $(-, a)$ is associated with the conclusion of any $0$-ary $?$-link of $\pi$: there is a ``sparing'' choice, that is to always choose $a = \multiset{}$. On the other hand, when a $0$-ary $?$-link has a conclusion which is the premise of a cut, one can never associate $(-, \multiset{})$ to this edge, since according to Definition~\ref{def:experiment} one cannot associate $(+, \multiset{})$ with the main conclusion of a $!$-link; nevertheless one can still make a ``sparing'' choice by choosing only a multiset of cardinality $1$.

\begin{definition}\label{def:w-sparing}
We define, by induction on $\depth{\pi}$, what it means to be \emph{$w$-sparing} for a $\smbis{}$-experiment $e$ of a net $\pi$:
\begin{itemize}
\item for every conclusion $c$ of a $0$-ary $?$-link of $\ground{\pi}$ which is not premise of some cut-link, we have $e(c) = (-, [])$;
\item for every conclusion $c$ of a $0$-ary $?$-link of $\ground{\pi}$ which is premise of some cut-link, we have $e(c) = (-, [\alpha])$ for some $\alpha \in D$;
\item for every $!$-link $o$ of $\ground{\pi}$, $e(o)$ is a finite multiset of $w$-sparing experiments of $\pi^o$.
\end{itemize}
\end{definition}

The following definition introduces an equivalence relation $\sim$ on the $\smbis{}$-experiments of a $\flat$-net $\pi$: intuitively $e\sim e'$ when $e$ and $e'$ associate with a given $!$-link of $\pi$ multisets of experiments with the same cardinality, and with the conclusion of a given $0$-ary $?$-link it can never happen that one of the two associates $(-, [])$ and the other one $(-, a)$ with $a \not= []$.

\begin{definition}\label{def:equivalenceSIM}
We define an \emph{equivalence $\sim$} on the set of $\smbis{}$-experiments of a $\flat$-net $\pi$, by induction on $\depth{\pi}$. Let $e,e':\pi$, we set $e \sim e'$ whenever
\begin{itemize}
\item for any weakening-link $l$ of $\ground{\pi}$, there is $m \in \integers$ such that $e(c)=(-, [\alpha_1, \ldots, \alpha_m])$ and $e'(c) = (-, [\alpha'_1, \ldots, \alpha'_m])$ for some $\alpha_1, \ldots, \alpha_m, \alpha'_1, \ldots, \alpha'_m \in D$, where $c$ is the conclusion of $l$;
\item and, for every $!$-node $o$ of $\ground{\pi}$, there is $m \in \integers$ such that $e(o) = [e_1, \ldots, e_m]$, $e'(o) = [e'_1, \ldots, e'_m]$ and, for any $j \in \{ 1, \ldots, m \}$, we have $e_j \sim e'_j$.
\end{itemize}
\end{definition}

We conclude the section by recalling the invariance of the $\sm{}$-interpretation w.r.t\ usual cut elimination, and by stating the invariance of the $\smbis{}$-interpretation w.r.t.\ non erasing cut elimination.

\begin{theorem}\label{th:invariance}
For $\pi$ and $\pi_1$ nets: if $\pi \cutred \pi_1$, then $\sm{\pi} = \sm{\pi_1}$.
\end{theorem}

\begin{proof}
See the proof of Theorem 11 p. 1891 of~\cite{CarvPagTdF10}, which is itself an adaptation of the original proof of~\cite{ll}.
\end{proof}

The newly defined $\smbis{}$-interpretation is invariant w.r.t.\ \emph{non erasing} cut elimination; the reader can check that $\smbis{}$-interpretation \emph{is not} invariant w.r.t.\ some erasing steps. We have the following proposition, which is an immediate consequence of Lemma~\ref{lemma : key-lemma : strat} of Section~\ref{sect:SN}:

\begin{proposition}\label{prop:invariancesmbis}
For $\pi$ and $\pi_1$ nets: if $\pi \nonerasingred \pi_1$, then $\smbis{\pi} = \smbis{\pi_1}$.
\end{proposition}

\section{Qualitative account}\label{sect:SN}

We present in this section our main qualitative result, contained in Corollary~\ref{corollary : cut strongly normalizable}. The first thing to notice here is that we cannot expect the exact analogue of the qualitative result proven in~\cite{CarvPagTdF10} on weak normalization (that is here recalled in Theorem~\ref{th:qualitatifWN}): this is because there exist nets $\pi$ and $\pi'$ such that $\sm{\pi}=\sm{\pi'}$ and $\pi$ is strongly normalizing while $\pi'$ is not, which clearly shows that there is no hope to extract the information on the strong normalizability of a net from its $\sm{}$-interpretation (see Remark~\ref{rem:NoSNsem} for a precise example of this phenomenon). We can nevertheless answer Question~\ref{uno1bis} raised in the introduction of the paper, thanks to the newly defined $\smbis{}$-interpretation, as follows:
\begin{itemize}
\item
we first prove that \emph{for a cut-free net} $\pi$ one can recover $\smbis{\pi}$ from $\sm{\pi}$ (Subsection~\ref{subsect:Sem+Sembis}, Proposition~\ref{prop:SembisFromSem})
\item
we then show how one can extract from the $\smbis{}$-interpretation the information that cannot be extracted from the $\sm{}$-interpretation (Subsection~\ref{subsect:SNsem}, Theorem~\ref{theorem:qualitativeSN})
\item
by combining the two previous points and starting from the two (good old) $\sm{}$-interpretations of two cut-free nets $\pi$ and $\pi'$, we can compute $\smbis{\pi}$ and $\smbis{\pi'}$, which allows to ``predict'' whether or not the net obtained by cutting $\pi$ and $\pi'$ is strongly normalizable (Corollary~\ref{corollary : cut strongly normalizable}).
\end{itemize}
In Subsection~\ref{subsection:conservation}, we give a variant of the standard proof of strong normalization for $MELL$ (\cite{ll,phddanos}). The interesting point is the alternative proof of the Conservation Theorem (here Theorem~\ref{theorem: conservation}), which is an immediate consequence of the qualitative results presented in Subsection~\ref{subsect:SNsem}.

\begin{figure}
\centering
\scalebox{\scalefact}{\input{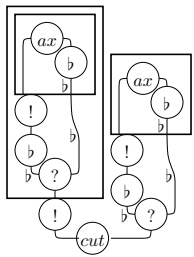}}
\hspace{-5pt}
\caption{Example of a non-normalizable net.}\label{fig:exanetnoWN}
\hrulefill
\end{figure}

\subsection{Two interpretations of nets}\label{subsect:Sem+Sembis}

Of course, the $\sm{}$-interpretation cannot characterize strongly normalizable nets, as the following remark shows.

\begin{remark}\label{rem:NoSNsem}
It is well-known that there are non-normalizable \emph{untyped} nets. A famous example is the net corresponding to the untyped $\lambda$-term $(\lambda x.xx)(\lambda x.xx)$ (see \cite{phddanos}, \cite{phdregnier}). We give in Figure~\ref{fig:exanetnoWN} a slight variant (which is not a $\lambda$-term), due to Mitsu Okada. The reader can check that this net reduces to itself by one $(!/?)$ step and one $(ax)$ step\footnote{This is not relevant for the purpose of this example, but notice that by Theorem~\ref{th:qualitatifWN} (proven in~\cite{CarvPagTdF10} and recalled in the following Subsection~\ref{subsect:SNsem}) the $\sm{}$-interpretation of the net in Figure~\ref{fig:exanetnoWN} is empty; a fact which can obviously be also checked directly on the net itself.}. 

Now, consider as net $\pi$ the net consisting of a unique $1$-link, and as net $\pi'$ the net of Figure~\ref{figure: pi'} consisting of a $1$-link and a $!$-link $o$ without auxiliary conclusions and having one main conclusion cut against the conclusion of a $0$-ary $?$-link, where the box $\pi^{o}$ is the net of Figure~\ref{fig:exanetnoWN} to which one adds (for example) a $1$-link, whose conclusion is the unique conclusion of $\pi^{o}$. The net $\pi$ is cut-free and thus strongly normalizable, while the net $\pi'$ is normalizable (just reduce the unique -erasing- cut-link of $\pi'$, which yields the net $\pi$), but not strongly normalizing since $\pi^{o}\cutred\pi^{o}$. On the other hand, clearly $\sm{\pi}=\sm{\pi'}$ (using Theorem~\ref{th:invariance} since $\pi'\onecut\pi$, but also by a straightforward computation one can check that $\sm{\pi}=\sm{\pi'}=\{(+,\ast)\}$).
\end{remark}

\begin{figure}
\centering
\scalebox{\scalefact}{\input{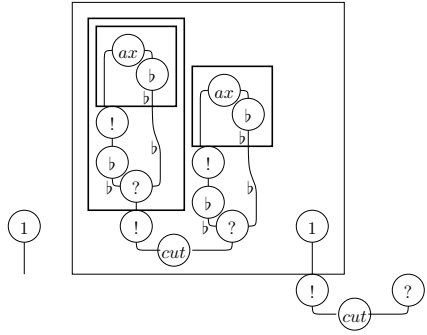}}
\hspace{-5pt}
\caption{The net $\pi'$ of Remark~\ref{rem:NoSNsem}: an exemple of normalizable net that is not strongly normalizable}\label{figure: pi'}
\hrulefill
\end{figure}

We use in the sequel the obvious notion of substitution, precisely defined as follows:

\begin{definition}[Substitution]\label{def:substitution}
A \emph{substitution} is a function $\sigma: D\rightarrow D$ induced by a function $\sigma^A : A \rightarrow D$ and defined by induction on the rank of the elements of $D$, as follows (as usual $ p \in \{ +, - \}$ and $a\in A$):
$$\begin{array}{rclcrclcrcl}
\sigma\seq{+,a} & \mathrel{:=} & \sigma^{A}(a) & & \sigma\seq{-,a} & \mathrel{:=} & {\sigma^{A}(a)}^\perp\\
\sigma\seq{p,\ast} & \mathrel{:=} & \seq{p,\ast} \\
\sigma\seq{p,x,y} & \mathrel{:=} & \seq{p,\sigma(x),\sigma(y)} & & \sigma\seq{p,[x_1,\dots,x_n]} & \mathrel{:=} & \seq{p,[\sigma(x_1),\dots,\sigma(x_n)]}
\end{array}$$
We denote by $ \mathcal{S}$ the set of substitutions. If $\mathbf{y} = \seq{x_1,\dots,x_n} \in D^n $, we set $\sigma(\mathbf{y})\mathrel{:=}\seq{\sigma(x_1),\dots,\sigma(x_n)}$.
\end{definition}

An immediate (but important) property mentioned in~\cite{CarvPagTdF10} is that the $\sm{}$-interpretation of a $\flat$-net is closed by substitution. This is still the case for the $\smbis{}$-interpretation of a $\flat$-net.

\begin{lemma}\label{lemma : closed by substitution}
Let $\pi$ be a $\flat$-net. For every $\smbis{}$-experiment $e'$ of $\pi$, for every $\sigma \in \mathcal{S}$, there is a $\smbis{}$-experiment $e$ of $\pi$ such that $(\sigma(\vert e' \vert), \sigma(\mathcal{W}(e'))) = (\vert e \vert, \mathcal{W}(e))$ and $e \sim e'$.
\end{lemma}

\begin{proof}
The proof is by induction on $\sizenet{\pi}$. In the two following cases:
\begin{itemize}
\item $\pi$ is an axiom
\item or in the ground-structure of $\pi$, there is a cut-link 
\end{itemize}
we use the property that, for any $x \in D$, for any $\sigma \in \mathcal{S}$, we have $\sigma(x^\perp) = \sigma(x)^\perp$. 

The other cases are trivial.
\end{proof}

We now define the function allowing to compute $\smbis{\pi}$ from $\sm{\pi}$, when $\pi$ is cut-free (Proposition~\ref{prop:SembisFromSem}). There are two simple ideas underlying the definition: 
\begin{itemize}
\item
since $\smbis{}$-experiments never associate the empty multiset of experiments to a $!$-link, we will never have $(+,\multiset{})\in\smbis{\pi}$, so that we can restrict to the exhaustive part of $\sm{\pi}$
\item
since $\smbis{}$-experiments allow to associate with the conclusion of a $0$-ary $?$-link $(-, a)$ for any $a \in \finitemultisets{D}$, to recover $\smbis{\pi}$ from $\sm{\pi}$ (actually from $\exhaustive{\sm{\pi}}$), we have to substitute in $\exhaustive{\sm{\pi}}$ every occurrence of $(-,\multiset{})$ with $(-, a)$ for all the possible $a \in \finitemultisets{D}$ (and of course we also have to keep track of those $a$ in $\mathcal{W}$).
\end{itemize}

\begin{definition}\label{def:F}
We define the function $F:(\exhaustive{D})^{n}\to\mathcal{P}_{f}(D^{n}\times\finitemultisets{D})$ by stating 
$$F(\sequence{x_{1},\ldots,x_{n}})= \{(\sequence{y_{1},\ldots,y_{n}}, \sum_{i=1}^n \mathcal{W}_i) \: ; \: (y_1, \mathcal{W}_1) \in F(x_{1}), \ldots, (y_n, \mathcal{W}_n) \in F(x_n) \}$$ 
and $F : \exhaustive{D} \to \mathcal{P}_{f}(D\times\finitemultisets{D})$\footnote{We keep the same notation for $F:(\exhaustive{D})^{n}\to\mathcal{P}_{f}(D^{n}\times\finitemultisets{D})$ and $F : \exhaustive{D} \to \mathcal{P}_{f}(D\times\finitemultisets{D})$.} is defined by induction on the rank of $x$\footnote{That is the least number $n\in\integers$ s.t. $x\in D_n$ (see Definition~\ref{def:D}).}:
\begin{minilist}
\item
if $x \in D_{0}$, then $F(x)=\{ (x, \multiset{}) \}$
\item
if $x = (\iota,y,y')$, then $F(x)=\{((\iota,z,z'), \mathcal{W} + \mathcal{W}') :\ (z, \mathcal{W}) \in F(y)\textrm{ and } (z', \mathcal{W}') \in F(y')\}$
\item if $x = (+, \beta)$ where $\beta = \multiset{x_{1}, \ldots, x_{k}} \in \finitemultisets{D^\textrm{ex}}$, then $F(x)=\{ ((+,[x'_{1},\ldots,x'_{k}]), \sum_{i=1}^k \mathcal{W}_i):\ (x'_{i}, \mathcal{W}_i)\in F(x_{i})\}$\footnote{Notice that since $x\in D^\textrm{ex}$, one has $k \geq 1$.}
\item if $x=(-,\beta)$ where $\beta=[x_{1},\ldots,x_{k}]\in\mathcal{M}_{\textrm{fin}}(D^\textrm{ex})$, then $F(x) = \{((-,[x'_{1},\ldots,x'_{k}]), \sum_{i=1}^k \mathcal{W}_i):\ (x'_{i}, \mathcal{W}_i) \in F(x_{i})\}$ if $k>0$ and $F(x)= \{ ((-, a), a) :\ a \in \finitemultisets{D} \}$ if $k=0$.
\end{minilist}
\end{definition}

An atomic experiment (see next Definition~\ref{definition:AtomicExperiments}) associates with every axiom link an element of $\{ +, - \} \times A$, and it is rather clear from Definition~\ref{def:experiment} that using the notion of substitution, from atomic experiments of $\pi$ one can recover any experiment of $\pi$. This remark can be shifted from experiments to points of the interpretation: by suitably defining (Definition~\ref{def:atomic}) the atomic part of the intepretation, one can recover $\sm{\pi}$ from $\atomic{\sm{\pi}}$ and $\smbis{\pi}$ from $\atomic{\smbis{\pi}}$. The notion of exhaustive $\sm{}$-experiment directly comes from~\cite{CarvPagTdF10}.

\begin{definition}\label{definition:AtomicExperiments}
For any net $\pi$, we define, by induction of $\depth{\pi}$, what means to be \emph{atomic} for any $\sm{}$-experiment (resp.\ $\smbis{}$-experiment) of $\pi$:
\begin{itemize}
\item 
An $\sm{}$-experiment (resp.\ $\smbis{}$-experiment) of a net $\pi$ of depth $0$ is said to be \emph{atomic} if it associates with every conclusion of every axiom of $\ground{\pi}$ an element of $\{ +, - \} \times A$.
\item 
An $\sm{}$-experiment (resp.\ $\smbis{}$-experiment) of a net $\pi$ of depth $n+1$ is said to be \emph{atomic} if 
\begin{itemize}
\item it associates with every conclusion of every axiom of $\ground{\pi}$ an element of $A$
\item and it associates with every $!$-link $o$ of $\ground{\pi}$ a finite multiset of atomic $\sm{}$-experiments (resp.\ $\smbis{}$-experiment) of $\pi^o$.
\end{itemize}
\end{itemize}

An $\sm{}$-experiment $e$ of a net $\pi$ is \emph{exhaustive} when $\result{e}\in(\exhaustive{D})^{n}$ for some $n\geq 0$.
\end{definition}

The following definition allows in particular to define the subset $\atomic{\sm{\pi}}$ of the ``atomic'' elements of $\sm{\pi}$, which will be used in Proposition~\ref{prop:SembisFromSem}. 

\begin{definition}\label{def:atomic}
Given $E \in \subsets{D^{n}}$ for some $n \geq 1$, we say that $r \in E$ is \emph{$E$-atomic} when for every $r'\in E$ and every substitution $\sigma$ such that $\sigma(r')=r$ one has $\sigma(\gamma)\in A$ for every $\gamma \in A$ that occurs in $r'$. For $E \in \subsets{D^{n}}$, we denote by ${E}_\textit{At}$ the subset of $E$ consisting of the $E$-atomic elements.
\end{definition}

For cut-free nets, the next lemma uses the ad hoc function introduced in Definition~\ref{def:F} to recover the atomic part of $\smbis{\pi}$ from the exhaustive atomic part of $\sm{\pi}$.

\begin{lemma}\label{lemma:SembisAtFromSemAt}
Let $\pi$ be a cut-free net. Then $\{ (\result{e}, \weakeningsofexperiment{e}) \: ; \: e \textrm{ is an atomic } \experimentbis{}{\pi} \} = \bigcup_{\textbf{x} \in \exhaustive{(\atomic{\sm{\pi}})}} F(\textbf{x})$.
\end{lemma}

\begin{proof}
To prove the inclusion $\bigcup_{\textbf{x} \in \exhaustive{(\atomic{\sm{\pi}})}} F(\textbf{x}) \subseteq \{ (\result{e}, \weakeningsofexperiment{e}) \: ; \: e \textrm{ is an atomic } \experimentbis{}{\pi} \}$, we prove, by induction on $\sizenet{\pi}$, that, for every exhaustive atomic $\sm{}$-experiment $e$ of $\pi$ and for every $(\textbf{y}, \mathcal{W}) \in F(\result{e})$, there exists an atomic $\smbis{}$-experiment $e'$ of $\pi$ such that $(\result{e'}, \mathcal{W}(e')) = (\textbf{y}, \mathcal{W})$.

To prove the inclusion $\{ (\result{e}, \weakeningsofexperiment{e}) \: ; \: e \textrm{ is an atomic } \experimentbis{}{\pi} \} \subseteq \bigcup_{\textbf{x} \in \exhaustive{(\atomic{\sm{\pi}})}} F(\textbf{x})$, we prove, by induction on $\sizenet{\pi}$, that, for every atomic $\smbis{}$-experiment $e'$ of $\pi$, there exists an exhaustive atomic $\sm{}$-experiment $e$ of $\pi$ such that $(\result{e'}, \mathcal{W}(e')) \in F(\result{e})$.
\end{proof}

\begin{proposition}\label{prop:SembisFromSem}
Let $\pi$ be a cut-free net. Then 
$\smbis{\pi} = \{ (\sigma(\textbf{y}), \sigma(\mathcal{W})) \: ; \: (\textbf{y}, \mathcal{W}) \in \bigcup_{\textbf{x} \in \exhaustive{(\atomic{\sm{\pi}})}} F(\textbf{x}) \textrm{ and } \sigma \in \mathcal{S} \}$.
\end{proposition}

\begin{proof}
One has $\atomic{\smbis{\pi}}=\{ (\result{e}, \weakeningsofexperiment{e}) \: ; \: e \textrm{ is an atomic } \experimentbis{}{\pi} \}$. Then apply  Lemma~\ref{lemma:SembisAtFromSemAt} and remember we already noticed that substitutions allow to recover $\smbis{\pi}$ from $\atomic{\smbis{\pi}}$.
\end{proof}

\subsection{Characterizing strong normalization}\label{subsect:SNsem}

Now that we know how to compute $\smbis{\pi}$ from $\sm{\pi}$ (in the cut-free case), we show how to use the $\smbis{}$-interpretation in order to characterize strong normalization.\\ 
In~\cite{CarvPagTdF10}, given an $\sm{}$-experiment $e$ of a net $\pi$, we defined the notion of size of $e$ (denoted there by $s(e)$); the following definition extends this notion to $\smbis{}$-experiments (writing $\sizeexperiment{e}$ instead of $s(e)$). We also introduce for an $\smbis{}$-experiment $e$ of a net $\pi$ a new notion of size (denoted by $\sizeexperimentbis{e}$), which is crucial to establish our main results (see Lemma~\ref{lemma : key-lemma : strat}).

\begin{definition}[Size of experiments]\label{def:experimentsize}
For every $\flat$-net $\pi$, for every $\smbis{}$-experiment $e$ of $\pi$, we define, by induction on $\textrm{depth}(\pi)$, the \emph{size of $e$}, $\sizeexperiment{e}$ for short, as follows:
\vspace{-2pt}
$$\sizeexperiment{e} =  \sizenet{\ground{\pi}} + \sum_{o \in ! (\ground{\pi})} \, \sum_{e^o \in e(o)} \sizeexperiment{e^o} \enspace .$$
\vspace{-4pt}
We set $\sizeexperimentbis{e} = \sizeexperiment{e} + 2 \card{\weakeningsofexperiment{e}}$.
\end{definition}

\begin{remark}\label{remark:sizebis}
We have 
\begin{eqnarray*}
& & \sizeexperimentbis{e} \\
& = & \sizenet{\ground{\pi}} + 2 \sum_{\begin{array}{c} c \textrm{ is the conclusion of a $?$-link of $\ground{\pi}$ with no premise} \\ e(c) = (-, \mu) \end{array}} \card{\mu} \\
& & + \sum_{o \in !(\ground{\pi})}  \sum_{e^o \in e(o)} \sizeexperimentbis{e^o}
\end{eqnarray*}
\end{remark}

\begin{definition}
A $1$-$\smbis{}$-experiment $e$ of a $\flat$-net $\pi$ is a $\smbis{}$-experiment such that, for any $!$-link $o$ of $\ground{\pi}$, we have $e(o) = [e_1]$ with $e_1$ a $1$-$\smbis{}$-experiment of the box $\pi^o$ of $o$. 
\end{definition}

\begin{remark}\label{remark:1-experiments}
When such a $1$-$\smbis{}$-experiment $e$ of $\pi$ exists, we have $\sizenet{\pi} = \sizeexperiment{e} = \min \{ \sizeexperiment{e} \: ; \: e \textrm{ is an $\smbis{}$-experiment of $\pi$} \}$.

Every cut-free net $\pi$ has a $1$-$\smbis{}$-experiment: any choice of pair $\{\sequence{+,x},\sequence{-,x^\perp}\}$ of elements of $D$ for the $ax$-nodes of $\pi$ induces a $1$-$\smbis{}$-experiment of $\pi$. Since an $\smbis{}$-experiment is allowed to associate with the conclusion of any $0$-ary $?$-node any element of $D$, even when $\pi$ is $\nonerasing$-normal there always exists a $1$-$\smbis{}$-experiment of $\pi$; and this $1$-$\smbis{}$-experiment can also be chosen $w$-sparing (Definition~\ref{def:w-sparing}).
\end{remark}

Since the size $\sizeexperiment{e}$ of an $\smbis{}$-experiment $e$ of $\pi$ depends only on $\pi$ and on the ``number of copies'' chosen for the boxes of $\pi$ (i.e. the cardinalities of the multisets associated recursively with the $!$-links - in particular, the $\smbis{}$-size of an experiment does not depend on its behaviour on the axioms), and since two equivalent $\smbis{}$-experiments of $\pi$ ``take the same number of copies'' for every box of $\pi$, two equivalent $\smbis{}$-experiments clearly have the same $\sm{}$-size. But also, when $e \sim e'$ one has $\card{\weakeningsofexperiment{e}} = \card{\weakeningsofexperiment{e'}}$\footnote{Notice that we do not have (in general) $\weakeningsofexperiment{e}=\weakeningsofexperiment{e'}$.}, so that eventually $\sizeexperimentbis{e} = \sizeexperimentbis{e'}$:

\begin{fact}\label{fact:equivExpMemeTaille}
If $e$ and $e'$ are two $\smbis{}$-experiments of a $\flat$-net $\pi$, then from $e \sim e'$ it follows that $\sizeexperimentbis{e} = \sizeexperimentbis{e'}$.
\end{fact}

We can now prove a crucial result which plays, in the framework of strong normalization, a similar role as the so-called ``Key-Lemma'' (Lemma 17 p.1893 and its variant Lemma 20 p.1896) of~\cite{CarvPagTdF10}.

\begin{lemma}\label{lemma : key-lemma : strat}
Let $\pi$ and $\pi_1$ be two nets such that $\pi \onenonerasing \pi_1$. Then 
\begin{enumerate}
\item for every $\experimentbis{e}{\pi}$, there exists an $\experimentbis{e_1}{\pi_1}$ such that $(\result{e}, \weakeningsofexperiment{e}) = (\result{e_1}, \weakeningsofexperiment{e_1})$ and 
$\sizeexperimentbis{e_1} < \sizeexperimentbis{e}$;
\item for every $\experimentbis{e_1}{\pi_1}$, there exists an $\experimentbis{e}{\pi}$ such that $(\result{e}, \weakeningsofexperiment{e}) = (\result{e_1}, \weakeningsofexperiment{e_1})$ and 
$\sizeexperimentbis{e_1} < \sizeexperimentbis{e}$.
\end{enumerate}

Moreover, if $\pi_1 = t(\pi)$ where $t$ a stratified non-erasing cut-link of $\pi$, then
\begin{itemize}
\item[1bis.]\label{item:1bis}
for every $\experimentbis{e}{\pi}$ such that $\sizeexperimentbis{e} = \min \{ \sizeexperimentbis{e} \: ; \: e \textrm{ is an } \experimentbis{}{\pi} \}$, there exists an $\experimentbis{e_1}{\pi_1}$ such that $(\result{e}, \weakeningsofexperiment{e}) = (\result{e_1}, \weakeningsofexperiment{e_1})$ and $\sizeexperimentbis{e_1} = \sizeexperimentbis{e} -2$;
\item[2bis.]\label{item:2bis}
for every $\experimentbis{e_1}{\pi_1}$ such that $\sizeexperimentbis{e_1} = \min \{ \sizeexperimentbis{e} \: ; \: e \textrm{ is an } \experimentbis{}{\pi_1} \}$, there exists an $\experimentbis{e}{\pi}$ such that $(\result{e}, \weakeningsofexperiment{e}) = (\result{e_1}, \weakeningsofexperiment{e_1})$ and $\sizeexperimentbis{e_1} = \sizeexperimentbis{e} -2$.
\end{itemize}
\end{lemma}

\begin{proof}
We first prove $1bis$ and $2bis$. By a straightforward adaptation of the proof given in~\cite{CarvPagTdF10}, one proves that if $t$ is a stratified non-erasing cut-link of $\pi$, then
\begin{itemize}
\item
for every $\smbis{}$-experiment $e$ of $\pi$ such that $\sizeexperimentbis{e} = \min \{ \sizeexperimentbis{e} \: ; \: e \textrm{ is an } \experimentbis{}{\pi} \}$, there exists an $\smbis{}$-experiment $e_1$ of $\pi_1$ such that $\result{e}=\result{e_1}$, $\weakeningsofexperiment{e}=\weakeningsofexperiment{e_1}$ and $\sizeexperiment{e_1} = \sizeexperiment{e} -2$;
\item
for every $\smbis{}$-experiment $e_1$ of $\pi_1$ such that $\sizeexperimentbis{e_1} = \min \{ \sizeexperimentbis{e} \: ; \: e \textrm{ is an } \experimentbis{}{\pi_1} \}$, there exists an $\smbis{}$-experiment $e$ of $\pi$ such that $\result{e}=\result{e_1}$, $\weakeningsofexperiment{e}=\weakeningsofexperiment{e_1}$ and $\sizeexperiment{e_1} = \sizeexperiment{e} -2$;
\end{itemize}
Furthermore, since the reduction step leading from $\pi$ to $\pi_1$ is non erasing, we have $\weakeningsofexperiment{e}=\weakeningsofexperiment{e_1}$, which yields $1bis$ and $2bis$.

The proof of $1$ and $2$ is by induction on $\depth{\pi}$. If $t$ is a cut-link at depth $0$, then $t$ is a stratified non-erasing cut-link of $\pi$, so we already know (by $1bis$ and $2bis$) that the properties hold. If $t$ is a cut-link of $\pi^o$ with $o \in ! (\ground{\pi})$, then, by induction hypothesis,
\begin{itemize}
\item[a.] 
for every $\experimentbis{e^{o}}{\pi^o}$, there exists an $\experimentbis{e^{o}_1}{t(\pi^o)}$ such that $(\result{e^{o}}, \weakeningsofexperiment{e^{o}}) = (\result{e^{o}_1}, \weakeningsofexperiment{e^{o}_1})$ and 
$\sizeexperimentbis{e^{o}_1} < \sizeexperimentbis{e^{o}}$; 
\item[b.] 
for every $\experimentbis{e^{o}_1}{t(\pi^{o})}$, there exists an $\experimentbis{e^{o}}{\pi^{o}}$ such that $(\result{e^{o}}, \weakeningsofexperiment{e^{o}}) = (\result{e^{o}_1}, \weakeningsofexperiment{e^{o}_1})$ and 
$\sizeexperimentbis{e^{o}_1} < \sizeexperimentbis{e^{o}}$.
\end{itemize}
Now, we can take $e$ and $e_1$ such that 
\begin{itemize}
\item for every edge $a$ of $\ground{\pi} = \ground{t(\pi)}$, we have $e(a) = e_1(a)$
\item for every $!$-link $o' \not= o$ of $\ground{\pi}= \ground{t(\pi)}$, $e(o') = e_1(o')$
\item $e(o)=[f^{o}_{1},\ldots,f^{o}_{k}]$, $e_{1}(o)=[f^{o}_{11},\ldots,f^{o}_{k1}]$, where $k\geq 1$ (remember Remark~\ref{rem:exp+expbis}) and $f^{o}_{i1}$ and $f^{o}_{i}$ are obtained by applying the induction hypothesis to $\pi^{o}$, following items $a.$ and $b.$, so that $\sizeexperimentbis{f^{o}_{i1}} < \sizeexperimentbis{f^{o}_{i}}$ for every $i\in\{1,\ldots,k\}$.
\end{itemize}
Thus, by Remark~\ref{remark:sizebis}, $\sizeexperimentbis{e_{1}} < \sizeexperimentbis{e}$.
\end{proof}

In order to precisely compare our results to the ones of~\cite{CarvPagTdF10}, we recall what is proven in~\cite{CarvPagTdF10}: in Theorem~\ref{th:qualitatifWN} and Corollary~\ref{cor:composition} we refer to ``head-normalization'' meaning stratified normalization at depth $0$.

\begin{theorem}\label{th:qualitatifWN}
Let $\pi$ be a net. We have:
\begin{minienum}
\item\label{cond:smpinon-empty} $\pi$ is head-normalizable iff $\sm{\pi}$ is non-empty;
\item\label{cond:smpiexnon-empty} $\pi$ is normalizable iff $\sm{\pi}^\textrm{ex}$ is non-empty.
\end{minienum}
\end{theorem}

Theorem~\ref{theorem:qualitativeSN} gives a characterization of strongly normalizable nets in terms of the $\smbis{}$-interpretation, which is very similar to the just recalled results for head-normalizable and (weakly) normalizable nets. Notice, however, that in general we cannot recover $\smbis{\pi}$ from $\sm{\pi}$, so that Theorem~\ref{theorem:qualitativeSN} itself cannot pretend to be a characterization of strongly normalizable nets in the relational model (remember by the way that strictly speaking this is not possible by Remark~\ref{rem:NoSNsem}).

\begin{proposition}\label{prop: WNstratnonerasing => smbis non-empty}
We have $\pi \in \textbf{WN}^{\nonerasing} \Rightarrow \smbis{\pi} \not= \emptyset$.
\end{proposition}

\begin{proof}
Let $\pi \nonerasingred \pi_{0}$ with $\pi_{0}$ $\nonerasing$-normal. There obviously exists a $1$-$\smbis{}$-experiment $e$ of $\pi_{0}$ (Remark~\ref{remark:1-experiments}), and thus $\result{e} \in \smbis{\pi}$ (by Proposition~\ref{prop:invariancesmbis}).
\end{proof}

\begin{theorem}\label{theorem:qualitativeSN}
A net $\pi$ is strongly normalizable iff $\smbis{\pi}$ is non-empty.
\end{theorem}

\begin{proof}
By Proposition~\ref{proposition:SN=SNnonerasing}, it is enough to show that, for any net $\pi$, we have $\pi \in \textbf{SN}^{\nonerasing}$ if, and only if, $\smbis{\pi}$ is non-empty.
If $\pi\in\textbf{SN}^{\nonerasing}$, then $\pi \in \textbf{WN}^{\nonerasing}$, hence we can apply Proposition~\ref{prop: WNstratnonerasing => smbis non-empty}.\\
Conversely, one proves by induction on $\min \{ \sizeexperimentbis{e} \: ; \: e \textrm{ is an } \experimentbis{}{\pi} \}$ that $\pi \in \textbf{SN}^{\nonerasing}$. If $\pi$ is $\nonerasing$-normal, we are done. Otherwise, we show that for every $\pi_{1}$ such that $\pi \onenonerasing \pi_{1}$, one has $\pi_{1} \in \textbf{SN}^{\nonerasing}$. Since $\smbis{\pi} \neq \emptyset$, there exist $\smbis{}$-experiments of $\pi$ and we can select $e$ such that $\sizeexperimentbis{e} = \min \{ \sizeexperimentbis{e} \: ; \: e \textrm{ is an } \experimentbis{}{\pi} \}$. By Lemma~\ref{lemma : key-lemma : strat}, there exists a $\smbis{}$-experiment $e_{1}$ of $\pi_{1}$ such that $\sizeexperimentbis{e_{1}} < \sizeexperimentbis{e}$, hence $\min \{ \sizeexperimentbis{e'} \: ; \: e' \textrm{ is an } \experimentbis{}{\pi_1} \} < \min \{ \sizeexperimentbis{e} \: ; \: e \textrm{ is an } \experimentbis{}{\pi} \}$:
 by induction hypothesis $\pi_{1} \in \textbf{SN}^{\nonerasing}$.
\end{proof}

An immediate consequence of Theorem~\ref{th:qualitatifWN} stated in~\cite{CarvPagTdF10} as Corollary 24 p.1897 is the following:

\begin{corollary}\label{cor:composition}
Let $\pi$ (resp.\ $\pi'$) be a net with conclusions $\mathbf{d}, c$ (resp.\ $\mathbf{d}', c'$). 
\begin{minienum}
\item 
The net $\seq{\pi | \pi'}_{c, c'}$ is head-normalizable iff there are $\mathbf{x}\in\sm{\pi}$ and $\mathbf{x}'\in\sm{\pi'}$ such that $\mathbf{x}_c = {\mathbf{x'}_{c'}}^\perp$.
\item 
The net $\seq{\pi | \pi'}_{c, c'}$ is normalizable iff there is $\mathbf{x}, \mathbf{x}'\in \mathbf{D}^{ex}$, $x \in D $ s.t.\ $ (\mathbf{x}, x) \in \sm{\pi}$ and $(\mathbf{x}', x^\perp)\in\sm{\pi'}$.
\end{minienum}
\end{corollary}

The following corollary, very much in the style of Corollary~\ref{cor:composition}, allows to answer Question~\ref{uno1bis} raised in the introduction, despite the fact that one cannot extract the information on the strong normalizability of a net from its $\sm{}$-interpretation: given two cut-free nets $\pi$ and $\pi'$, thanks to Proposition~\ref{prop:SembisFromSem} we can compute $\smbis{\pi}$ (resp.\ $\smbis{\pi'}$) from $\sm{\pi}$ (resp.\ $\sm{\pi'}$), and the corollary allows then to ``predict'' (by purely semantic means) whether or not the net obtained by cutting $\pi$ and $\pi'$ is strongly normalizing.

\begin{corollary}\label{corollary : cut strongly normalizable}
Let $\pi$ (resp. $\pi'$) be a net with conclusions $\textbf{d}, c$ (resp. $\mathbf{d'}, c'$). The net $\cutnets{\pi}{\pi'}{c}{c'}$ is strongly normalizable if, and only if, there are $(\mathbf{x}, \mathcal{W}) \in \smbis{\pi}$ and $(\mathbf{x'}, \mathcal{W'}) \in \smbis{\pi'}$ such that $\mathbf{x}_c = {\mathbf{x'}_{c'}}^\perp$.
\end{corollary}

The reader should notice that considering only nets of the form $\cutnets{\pi}{\pi'}{c}{c'}$ with $\pi$ and $\pi'$ cut-free might look as a restriction, but we already noticed in~\cite{CarvPagTdF10} that this is not quite true, as the following proposition (which is a variant of Proposition 34 p. 1899 of~\cite{CarvPagTdF10}) shows:

\begin{proposition}\label{prop:generalbound}
For every net $\pi_1$ with conclusions $\mathbf{d}$, there exist two cut-free nets $\pi$ and $\pi'$ with conclusions resp.\ $\mathbf{d}$, $c$ and $c'$ such that: 
\begin{minienum}
 \item\label{item1prop:generalbound}
$\seq{\pi | \pi'}_{c, c'}\cutred\pi_1$;
\item\label{item2prop:generalbound}
$\pi_1\in\textbf{SN}$ iff $\seq{\pi | \pi'}_{c, c'}\in\textbf{SN}$, and we have  $\strong{\pi_1}\leq\strong{\seq{\pi | \pi'}_{c, c'}}\enspace.$
\end{minienum}
\end{proposition}

\begin{proof}
See the proof of Proposition 34 p.1899 of~\cite{CarvPagTdF10} for the definition of $\pi$ and $\pi'$, where it is also proven that $\seq{\pi | \pi'}_{c, c'}\cutred\pi_1$. The fact that $\pi_1\in\textbf{SN}$ iff $\seq{\pi | \pi'}_{c, c'}\in\textbf{SN}$ is immediate from the definition of $\pi$ and $\pi'$, and the fact that $\strong{\pi_1}\leq\strong{\seq{\pi | \pi'}_{c, c'}}$ is obvious.
\end{proof}

\begin{remark}\label{remark: conservation}
Proposition~\ref{prop: WNstratnonerasing => smbis non-empty} and Theorem~\ref{theorem:qualitativeSN} together give a new proof of the following theorem for the nets of Definition~\ref{def:struct}:
\begin{theorem}[Conservation Theorem] \label{theorem: conservation}
We have $\textbf{WN}^{\nonerasing} = \textbf{SN}$.
\end{theorem}
As a corollary, we can show, for instance, that any MELL typed net is strongly normalizing only by showing that any MELL net is in $\textbf{WN}^{\nonerasing}$. This is done in Subsection~\ref{subsection:conservation}.
\end{remark}


\subsection{Strong Normalization for MELL nets}\label{subsection:conservation}


A  (typeable) $MELL$ net is a net of Definition~\ref{def:struct}, where with every logical edge one can associate a formula of the logical language (we say this formula is ``the type'' of the edge), and the standard conditions on formulas have to be satisfied (see~\cite{ll} or any more recent reformulation like~\cite{phdtortora}). Recall the grammar of $MELL$ formulas:

$$A::=1\ |\ \bot\ |\ X\ |\ X^{\bot}\ |\ A\otimes A\ |\ A\lpar A\ |\ ?A\ |\ !A$$

where $X$ ranges over a set of propositional variables.

Notice that the constraint on the types of the edges imply that an $MELL$ net can never contain a clash: in the whole section, every net is clash-free.

\begin{definition}[multiset ordering]\label{def:MultisetOrder}
If $X$ is a set and $m\in \finitemultisets{X}$, recall that we denote by $\Supp{m}$ the set underlying $m$, and for $x\in X$, we denote by $m(x)$ the multiplicity of $x$ in the multiset $m$. A binary relation $<$ on $X$ induces a binary relation (still denoted by $<$ in this paper) on $\finitemultisets{X}$: for $m,m'\in \finitemultisets{X}$, one defines, by induction on $\card{\Supp{m}}$, when $m<m'$ holds:
\begin{itemize}
\item if $m = m'=\emptyset$, we do not have $m<m'$; 
\item
 if $m=\emptyset$ and $m'\neq\emptyset$, then $m<m'$ (and we do not have $m'<m$); 
 \item otherwise $m\neq\emptyset$, $m'\neq\emptyset$, and we have $m <m'$ iff one of the following holds, where $M= \max (\Supp{m})$ and $M'= \max (\Supp{m'})$:
\begin{itemize}
\item
$M<M'$
\item
$M=M'$ and $m(M)<m'(M')$
\item
$M=M'$, $m(M)=m'(M')$ and $m_{1} < m'_{1}$, where $m_{1}\in \finitemultisets{X}$, $\Supp{m_{1}}=\Supp{m} \setminus \{ M \}$ and for every $x\in\Supp{m_{1}}$ one has $m_{1}(x)=m(x)$ (resp.\ where $m'_{1}\in \finitemultisets{X}$, $\Supp{m'_{1}}=\Supp{m'}\setminus \{ M \}$ and for every $x \in \Supp{m'_{1}}$ one has $m'_{1}(x)=m'(x)$).
\end{itemize}
\end{itemize}
\end{definition}

\begin{remark}\label{rem:MultisetisWF}
This definition is equivalent to the definition given in \cite{dershowitz-manna} and it is well-known that when $<$ is well-founded on $X$, so is also the induced relation on $\finitemultisets{X}$. In particular, if $(\integers,<)$ is the set of natural numbers with the usual order relation, the ordered set $(\finitemultisets{\integers}, <)$ is well-founded and we can thus prove properties by induction on the multiset order relation on $\finitemultisets{\integers}$.
\end{remark}

\begin{definition}\label{def:sizes}
The complexity of a $MELL$ formula $A$ (notation $\compl{A}$) is the number of occurrences of logical operators (meaning the symbols $1,\bot,\otimes,\lpar,?,!$) occurring in $A$.

Let $\pi$ be an $MELL$ net:
\begin{itemize}
\item
a cut-node of type $(!/?)$ is \emph{linear} when the $?$-node whose conclusion is a premise of the cut has a unique premise
\item
the \emph{cut-size} of $\pi$ (notation $\Cut{\pi}$) is the multiset of natural numbers such that $\Supp{\Cut{\pi}}=\{\compl{A}:\textrm{ $A$ and $A^{\bot}$ are the types of the premises of a cut-node of }\pi\}$, and if $n\in\Supp{\Cut{\pi}}$, then $\Cut{\pi}(n)$ is the number of cut-nodes of $\pi$ whose premises have types with complexity $n$.
\end{itemize}
\end{definition}

\begin{remark}
Notice that $\compl{A}=\compl{A^{\bot}}$ for any $MELL$ formula $A$, so that the types of the two premises of any cut-node always have the same complexity. We will thus speak in the sequel of the complexity of a cut-node, meaning the complexity of any of its premises.
\end{remark}

\begin{lemma}\label{lemma:CutSizeShrinks}
Let $\pi$ be a net and $t$ be a non-erasing cut-link of $\pi$ such that one of the following holds:
\begin{itemize}
\item
$t$ is not of type $(!/?)$ or $t$ is linear
\item
$t$ is a non linear $(!/?)$ cut-node and $\pi^{o}$ is cut-free, where $o$ is the $!$-link whose main conclusion is a premise of $t$ and $\pi^{o}$ is the box of $o$.
\end{itemize}
Then $\Cut{t(\pi)}<\Cut{\pi}$, following the multiset ordering of Definition~\ref{def:MultisetOrder}.
\end{lemma}

\begin{proof}
If $t$ is not of type $(!/?)$ or $t$ is linear, it is obvious, following Definition~\ref{def:Scutreduction}, that every cut-node of $\pi$ different from $t$ appears unchanged in $t(\pi)$ and that $t$ ``becomes'' one or more cuts, but in any case all these cuts have complexity strictly smaller than the one of $t$: $\Cut{t(\pi)}<\Cut{\pi}$.

If $t$ is a non linear $(!/?)$ cut-node such that $\pi^{o}$ is cut-free, recalling Figure~\ref{fig:cut} one can see that $t$ ``becomes'' $k\geq 2$ cuts with complexity strictly smaller than the complexity of $t$. Concerning the other cut-nodes, again it is obvious that a cut-node of $\pi$ different from $t$ \emph{which does not occur in $\pi^{o}$} appears unchanged in $t(\pi)$. Now we can apply the crucial hypothesis that $\pi^{o}$ is cut-free: the nodes of $\pi^{o}$ appear several times in $t(\pi)$, but none of them is a cut-node. Then $\Cut{t(\pi)}<\Cut{\pi}$.
\end{proof}

\begin{proposition}\label{prop:WnforMELL}
If $\pi$ is an $MELL$ net, then $\pi\in\textbf{WN}^{\nonerasing}$.
\end{proposition}

\begin{proof}
It is an immediate consequence of Lemma~\ref{lemma:CutSizeShrinks} and of the following observation: if $\pi$ contains a non-erasing cut-node, then there exists a cut-node $t$ of $\pi$ satisfying the hypothesis of Lemma~\ref{lemma:CutSizeShrinks}. Indeed, either there exists in $\pi$ a linear cut-node or a cut-node which is not of type $(!/?)$, and we are done. Or every cut-node of $\pi$ is a non linear $(!/?)$ cut-node, in which case there exists a $!$-link $o$ of $\pi$ whose main conclusion is a premise of a cut-node $t$ and such that its box $\pi^{o}$ is cut-free.

More precisely, the proof is by induction on $\Cut{\pi}$. If $\pi$ is $\nonerasing$-normal the conclusion is immediate. Otherwise, by the previous observation, there exists a cut-node $t$ of $\pi$ satisfying the hypothesis of Lemma~\ref{lemma:CutSizeShrinks}. We thus have $\Cut{t(\pi)}<\Cut{\pi}$ and we can apply the induction hypothesis to $t(\pi)$: from $t(\pi)\in\textbf{WN}^{\nonerasing}$ it follows that $\pi\in\textbf{WN}^{\nonerasing}$.
\end{proof}

\begin{remark}
It is immediate to extend Lemma~\ref{lemma:CutSizeShrinks} to the case of erasing cuts, so that the proof of Proposition~\ref{prop:WnforMELL} becomes a (very easy) proof of weak normalization for $MELL$.
\end{remark}


\begin{corollary}\label{corollary:SN}
Every $MELL$ net is strongly normalizable.
\end{corollary}

\begin{proof}
Apply Proposition~\ref{prop:WnforMELL} and Theorem~\ref{theorem: conservation}.
\end{proof}

\begin{remark}\label{rem:StabilitySem=Conservation}
The proof of strong normalization for linear logic or for any of its remarkable fragments is usually split in two parts: weak normalization and a conservation theorem (see~\cite{ll},\cite{phddanos},\cite{SNLL10}), relying on a confluence result. The only strong normalization proofs we know for (fragments of) linear logic that do not use confluence are by Joinet (~\cite{phdjoinet}) and Accattoli (\cite{Accattoli13}). Our proof follows the traditional pattern (weak normalization+conservation theorem), but the conservation theorem (whose proof is usually very delicate: see~\cite{phddanos},\cite{SNLL10}) is here an immediate consequence of our ``semantic'' approach. In particular, our proof does not rely on confluence.
\end{remark}

\section{Quantitative account}\label{sect:SNquantitative}

In this section, we answer Question~\ref{due2bis} raised in the introduction: the point is to compute $\strong{\seq{\pi | \pi'}_{c, c'}}$ from $\sm{\pi}$ and $\sm{\pi'}$ with $\pi$ and $\pi'$ cut-free. By Proposition~\ref{prop:SembisFromSem}, we can substitute $\smbis{\pi}$ and $\smbis{\pi'}$ for $\sm{\pi}$ and $\sm{\pi'}$. On the other hand, by Corollary~\ref{corollary:postponingerasing}, we know that (provided $\cutnets{\pi}{\pi'}{c}{c'}$ is strongly normalizable) there exists $R_1 : \cutnets{\pi}{\pi'}{c}{c'} \stratnonerasingred \pi_1$ and $R_2 : \pi_1 \erasingred \pi_2$ antistratified, such that $\pi_1$ is $\nonerasing$-normal, $\pi_2$ is cut-free and $\strong{\cutnets{\pi}{\pi'}{c}{c'}} = \length{R_1} + \length{R_2}$. Summing up, in order to answer our question, we can compute $\length{R_1}$ and $\length{R_2}$ from $\smbis{\pi}$ and $\smbis{\pi'}$\footnote{Notice that since $\smbis{\cutnets{\pi}{\pi'}{c}{c'}}$ can be easily obtained from $\smbis{\pi}$ and $\smbis{\pi'}$, we can also freely use $\smbis{\cutnets{\pi}{\pi'}{c}{c'}}$.}

An important step in the computation of $\length{R_1}$ is the passage through experiments of $\cutnets{\pi}{\pi'}{c}{c'}$: we prove in Proposition~\ref{proposition : non-erasing stratified reduction} that $\length{R_1}$ can be expressed in terms of $\sizeexperiment{e}$, where $e$ is an $\smbis{}$-experiment of $\cutnets{\pi}{\pi'}{c}{c'}$ with minumum size. In the proof of Theorem~\ref{theorem:exactSN}, we show how 
$\sizeexperiment{e}$ can be obtained from suitable points of $\smbis{\pi}$ and $\smbis{\pi'}$.

Concerning $\length{R_2}$, notice that if we could know the exact number of (erasing) cut-links of $\pi_{1}$, we would also know $\length{R_2}$: these two numbers coincide, since obviously the length of any antistratified reduction sequence starting from a $\nonerasing$-normal net and leading to a cut-free net is the number of cuts of the $\nonerasing$-normal net. We thus compute the number of cut-links of $\pi_{1}$ in Lemma~\ref{lemma : erasing cuts}: it is the second component of $(\result{e_{1}},\mathcal{W}(e_{1}))\in\smbis{\cutnets{\pi}{\pi'}{c}{c'}}$, where $\sizeexperimentbis{e_1} = \min \{ \sizeexperimentbis{e} \: ; \: e \textrm{ is an $\smbis{}$-experiment of $\pi_{1}$} \}$. Lemma~\ref{lemma : key-lemma : strat} allows then to conclude that $(\result{e_{1}},\mathcal{W}(e_{1}))=(\result{e_{0}},\mathcal{W}(e_{0}))$, where $\sizeexperimentbis{e_0} = \min \{ \sizeexperimentbis{e} \: ; \: e \textrm{ is an $\smbis{}$-experiment of $\cutnets{\pi}{\pi'}{c}{c'}$} \}$. In the proof of Theorem~\ref{theorem:exactSN}, we explain how to select $(x,\mathcal{W})\in\smbis{\cutnets{\pi}{\pi'}{c}{c'}}$ so that $(x,\mathcal{W})=(\result{e_{0}},\mathcal{W}(e_{0}))$.

\begin{lemma}\label{lemma : erasing cuts}
Let $\pi$ be a $\nonerasing$-normal net. 
Let $e_0$ be an $\smbis{}$-experiment of $\pi$ such that 
$$\sizeexperimentbis{e_0} = \min \{ \sizeexperimentbis{e} \: ; \: e \textrm{ is an $\smbis{}$-experiment of $\pi$} \}.$$ 
Then $\card{\weakeningsofexperiment{e_0}}$ is the number of cuts of $\pi$.
\end{lemma}

\begin{proof}
Given a $\nonerasing$-normal net $\pi$, if there exists a $w$-sparing $1$-$\smbis{}$-experiment $e_1$ of $\pi$, then 
\begin{itemize}
\item any $\smbis{}$-experiment $e_0$ such that $\sizeexperimentbis{e_0} = \min \{ \sizeexperimentbis{e} \: ; \: e \textrm{ is an $\smbis{}$-experiment of $\pi$} \}$ is a $w$-sparing $1$-$\smbis{}$-experiment
\item and $\card{\weakeningsofexperiment{e_1}}$ is the number of cuts of $\pi$.
\end{itemize}
To conclude, notice that there always exists a $w$-sparing $1$-$\smbis{}$-experiment of a $\nonerasing$-normal net (Remark~\ref{remark:1-experiments}).
\end{proof}

We now need a notion of size of an element of the $\smbis{}$-interpretation of a net, which is a particular case of size of an element of $D^{n}\times\finitemultisets{D}$. Like for the notion of size of an experiment, we use the notion of size of an element of $D$ introduced in~\cite{CarvPagTdF10}:

\begin{definition}[Size of elements]\label{def:sizeelement}
For every $x\in D$, we define the \emph{size $\sizepoint{x}$ of $x$}, by induction on $\rank{x}$. Let $p \in \{+,-\}$,
\begin{minilist}
\item if $x \in \{ +, - \} \times A$ or $x= (p, \ast)$, then $\sizepoint{x}=1$;
\item if $x = (p, y, z)$, then $\sizepoint{x} = 1 + \sizepoint{y} + \sizepoint{z}$;
\item if $x = (p, [x_1, \ldots, x_m])$, then $ \sizepoint{x} = 1 + \sum_{j=1}^m \sizepoint{x_{j}}$;
\end{minilist}

\bigskip

Given $(x_1, \ldots, x_n) \in D^n$ ($n\geq 0$), we set $ \sizepoint{x_1, \ldots, x_n} = \sum_{i=1}^n \sizepoint{x_{i}}$ and $\sizepoint{[x_1, \ldots, x_n]} = \sum_{i=1}^n \sizepoint{x_{i}}$  

Let $n\geq 1$ and $(\textbf{x}, \mathcal{W}) \in D^{n} \times \finitemultisets{D}$. Then we set 
$\sizepointbis{\textbf{x}, \mathcal{W}} = \sizepoint{\textbf{x}} + \sum_{\alpha \in D} \mathcal{W}(\alpha) \cdot (\sizepoint{\alpha} + 2)$.
\end{definition}

\begin{remark}
Notice that for every point $x \in D$ or $x \in \finitesequences{D} \cup \finitemultisets{D}$, $\sizepoint{x}$ is the number of occurrences of $+$, $-$ in $x$ (seen as a word).
\end{remark}

\begin{definition}
Let $n \geq 1$. For any $X \subseteq D^{n} \times \finitemultisets{D}$, we set $\sizebisinf{X} = \inf \{ \sizepointbis{x} ; x \in X \} \in \integers \cup \{ \infty \}$.
\end{definition}

\begin{lemma}\label{lemma : sbis_inf(pi)}
Let $\pi$ be a $\flat$-net with $k$ structural conclusions. 
If $\pi$ is $\nonerasing$-normal, then we have $\sizebisinf{\smbis{\pi}} = \sizenet{\pi} + k = \min \{ \sizeexperiment{e} \: ; \: e \textrm{ is an $\smbis{}$-experiment of $\pi$} \} +k$.
\end{lemma}

\begin{proof}
We consider the cut-free net $\pi'$ obtained from $\pi$ in two steps:
\begin{itemize}
\item first, we erase all the weakening-links premises of some cut-link and all the cut-links;
\item second, under every $!$-link whose conclusion was premise of some cut-link, we add a $\flat$-link and a unary $?$-link at depth $0$ under this $\flat$-link.
\end{itemize}
First, notice that we have $\sizenet{\pi'} = \sizenet{\pi}$. Second, notice that, for any $w$-sparing $1$-$\smbis{}$-experiment $e$ of $\pi$, the $1$-experiment $e'$ of $\pi'$ induced by $e$\footnote{Notice that $e'$ is both an $\smbis{}$-experiment and an $\sm{}$-experiment of $\pi'$.} enjoys the following property: 
$\sizepoint{\result{e'}} = \sizepointbis{\result{e}, \mathcal{W}(e)}$.

Now, since $\pi$ is $\nonerasing$-normal, we can define, by induction on $\depth{\pi}$, a $w$-sparing $1-\smbis{}$-experiment $e_1 : \pi$ that associates $(p, \ast)$ with the conclusions of axiom nodes. More precisely, $e_1$ is defined as follows: 
\begin{itemize}
\item with every conclusion of a weakening of $\ground{\pi}$ that is premise of some cut, $e_1$ associates the element $(-, [\alpha^\perp])$, where $\alpha$ is such that $e_1$ associates $(+, [\alpha])$ with the other premise of the cut;
\item with every pair of conclusions of every $ax$-link of $\ground{\pi}$, $e_1$ associates the pair of elements $(+, \ast)$, $(-, \ast)$ (it does not matter in which order);
\item with every $!$-link $o$, $e_1$ associates the singleton $[e_1^o]$, where $e_1^o$ is an experiment defined as $e_1$ on $\pi^o$ (notice that $\depth{\pi^o}<\depth{\pi}$).
\end{itemize}
We denote by $e'_1$ the $1$-experiment of $\pi'$ induced by $e_1$: we have $\sizepoint{\result{e_1'}} = \sizeexperiment{\pi'} + k$ (induction on $\depth{\pi'}$) and $\sizepoint{\result{e'_1}} = \sizepointbis{\result{e_1}, \mathcal{W}(e_1)}$, hence 
$\sizenet{\pi} + k = \sizenet{\pi'} + k = 
\sizepointbis{\result{e_1}, \mathcal{W}(e_1)}$. By Remark~\ref{remark:1-experiments}, we have $\sizenet{\pi} = \sizeexperiment{e_1} = \min \{ \sizeexperiment{e} \: ; \: e \textrm{ is an $\smbis{}$-experiment of $\pi$} \}$. Lastly, since $e_1$ is a $w$-sparing atomic $1$-experiment of $\pi$ that associates $(p, \ast)$ with the conclusions of axiom nodes, we have $\sizepointbis{\result{e_1}, \mathcal{W}(e_1)} = \sizebisinf{\smbis{\pi}}$.
\end{proof}

We can now compute the length of $R_1$ by means of experiments; this is of course only a first step, since (still keeping the notations of Corollary~\ref{corollary:postponingerasing}) we are only allowed to use the elements of $\smbis{\cutnets{\pi}{\pi'}{c}{c'}}$ and not the experiments that produce these elements.

\begin{proposition}\label{proposition : non-erasing stratified reduction}
Let $\pi$ be a net and let $\pi'$ be a $\nonerasing{}$-normal net. For every reduction sequence $R: \pi \stratnonerasingred \pi'$, and every $\smbis{}$-experiment $e_0$ of $\pi $ such that $\sizeexperimentbis{e_0} = \min \{ \sizeexperimentbis{e} \: ; \: e \textrm{ is an $\smbis{}$-experiment of $\pi$} \}$,  we have $ \length{R} = (\sizeexperiment{e_0} - \sizebisinf{\smbis{\pi}})/2 $.
\end{proposition}

\begin{proof}
By induction on $\length{R}$. If $\length{R} = 0$, apply Lemma~\ref{lemma : sbis_inf(pi)}.

Now, $R = \pi \stratnonerasingred \pi_1 \stratnonerasingred \pi'$. By Lemma~\ref{lemma : key-lemma : strat}, there is an $\smbis{}$-experiment $e_1$ of $\pi_1$ such that $(\result{e_1}, \weakeningsofexperiment{e_1}) = (\result{e_0}, \weakeningsofexperiment{e_0})$, $\sizeexperimentbis{e_1} = \sizeexperimentbis{e_0} - 2$ and $\sizeexperimentbis{e_1} = \min \{ \sizeexperimentbis{e} \: ; \: e \textrm{ is an $\smbis{}$-experiment of $\pi_{1}$} \}$.

We have $\sizeexperiment{e_0} - \sizeexperiment{e_1} = \sizeexperimentbis{e_0} - \sizeexperimentbis{e_1} = 2$.

We apply the induction hypothesis to $\pi_1$. We have $\length{R} - 1 = (\sizeexperiment{e_1} - \sizebisinf{\smbis{\pi_1}}) / 2 = (\sizeexperiment{e_1} - \sizebisinf{\smbis{\pi}}) / 2 = (\sizeexperiment{e_0} - 2 - \sizebisinf{\smbis{\pi}}) / 2$
\end{proof}

The following lemma shows that if $\pi$ is cut-free and has no structural conclusions and $e$ is an $\smbis{}$-experiment of $\pi$, then $\sizeexperiment{e} \leq \sizepoint{\result{e}} - \sizepoint{\mathcal{W}(e)}$:

\begin{lemma}\label{lemma:sizeexperiment <= sizepoint}
Let $\pi$ be a cut-free $\flat$-net with $k$ structural conclusions (and possibly other logical conclusions) and let $e$ be an $\smbis{}$-experiment of $\pi$. Then we have $\sizeexperiment{e} \leq \sizepoint{\result{e}} - \sizepoint{\mathcal{W}(e)} - k$.
\end{lemma}

\begin{proof}
The proof is by induction on $\sizeexperiment{\pi}$. If $\ground{\pi}$ is an axiom, then $k = 0$ and $\sizepoint{\mathcal{W}(e)} = 0$: if the elements of $D$ associated with the conclusions of the axiom are of the shape $(p, a)$ with $a \in A \cup \{ \ast \}$, then we have $\sizeexperiment{e} = \sizepoint{\result{e}}$; else, we have $\sizeexperiment{e} < \sizepoint{\result{e}}$. Now, assume that $\ground{\pi}$ is a $!$-link $o$ with $k$ structural conclusions. Set $e(o) = [e_1, \ldots, e_m]$ with $m \geq 1$ and let $\pi^o$ be the box of $o$. Notice that $\pi$ has $k+1$ conclusions. We have
\begin{eqnarray*}
\sizeexperiment{e} & = & 1 + \sum_{j=1}^m \sizeexperiment{e_j} \\
& \leq & 1 + \sum_{j=1}^m (\sizeexperiment{\result{e_j}} - \sizepoint{\mathcal{W}(e_j)} -k) \qquad\quad\textrm{(by induction hypothesis)}\\
& = & 1 + \sizepoint{\vert e \vert} - \sizepoint{\mathcal{W}(e)} - (k+1)\\
& = & \sizepoint{\result{e}} - \sizepoint{\mathcal{W}(e)} - k .
\end{eqnarray*}
The other cases are left to the reader.
\end{proof}

Provided the set of atoms $A$ is infinite, if the size $\sizeexperiment{e}$ of the experiment $e$ does not reach the bound of Lemma~\ref{lemma:sizeexperiment <= sizepoint}, one can always choose a representative of the $\sim$-equivalence class of $e$ whose size does reach the bound. More precisely:

\begin{lemma}\label{lemma:sizepoint atteint}
Assume $A$ is infinite. Let $\pi$ be a cut-free $\flat$-net with $k$ structural conclusions (and possibly other logical conclusions), and let $e$ be an $\smbis{}$-experiment of $\pi$. There exist $e' \sim e$ and a substitution $\sigma$ such that $\sizeexperiment{e'} = \sizepoint{\result{e'}} - \sizepoint{\mathcal{W}(e')} - k$ and $\sigma(\result{e'}, \mathcal{W}(e')) = (\result{e}, \mathcal{W}(e))$.
\end{lemma}

\begin{proof}
Let $A_0$ be the set of elements of $A$ occurring in $\mathcal{W}(e)$. 
We prove, by induction on $\sizenet{\pi}$, that, for every infinite subset $A'$ of $A \setminus A_0$, there is an experiment $e' \sim e$ such that
\begin{enumerate}
\item $\sizeexperiment{e'} = \sizepoint{\result{e'}} - \sizepoint{\mathcal{W}(e')} -k$;
\item $\sigma(\result{e'}, \mathcal{W}(e')) = (\result{e}, \mathcal{W}(e))$ for some $\sigma \in \mathcal{S}$ such that $\restriction{\sigma}{A_0} = id_{A_0}$;
\item and every element of $A \setminus A_0$ occurring in $\result{e'}$ is an element of $A'$.
\end{enumerate}
In the case $\ground{\pi}$ is a weakening-link $l$, we set $e'(c) = e(c)$, where $c$ is $l$'s conclusion. The other cases are similar to the proof of Lemma~35 of \cite{CarvPagTdF10}.
\end{proof}

In order to prove our quantitative result (Theorem~\ref{theorem:exactSN}), we start relating, for $\smbis{}$-experiments $e$,  $\sizeexperimentbis{e}$ to the size of suitable elements of $\smbis{\pi}$.

\begin{lemma}\label{lemma:s(e)andequivalence}
Assume $A$ is infinite. Let $\pi$ be a cut-free net and let $e$ be an $\smbis{}$-experiment of $\pi $. We have
$\sizeexperimentbis{e} = \min \{ \sizepoint{\result{e'}} - \sizepoint{\mathcal{W}(e')} + 2 \card{\mathcal{W}(e')} ; \: e' \sim e \textrm{ and } (\exists \sigma \in \mathcal{S}) \sigma(\result{e'}, \mathcal{W}(e')) = (\result{e}, \mathcal{W}(e)) \}.$
\end{lemma}

\begin{proof}
We set $q = \min \{ \sizepoint{\result{e'}} - \sizepoint{\mathcal{W}(e')} + 2 \card{\mathcal{W}(e')} ; \: e' \sim e \textrm{ and } (\exists \sigma \in \mathcal{S}) \sigma(\result{e'}, \mathcal{W}(e')) = (\result{e}, \mathcal{W}(e)) \}.$

First, we prove $\sizeexperimentbis{e} \leq q$. Let $e'_0$ be an $\smbis{}$-experiment of $\pi $ such that $e'_0 \sim e$ and $\sizepoint{\result{e'_0}} - \sizepoint{\mathcal{W}(e'_0)} + 2 \card{\mathcal{W}(e'_0)} = q$. By Fact~\ref{fact:equivExpMemeTaille} and Lemma~\ref{lemma:sizeexperiment <= sizepoint}, we have $\sizeexperimentbis{e} = \sizeexperimentbis{e'_0} = \sizeexperiment{e'_0} + 2 \card{\mathcal{W}(e'_0)} \leq \sizepoint{\result{e'_0}} - \sizepoint{\mathcal{W}(e'_0)} + 2 \card{\mathcal{W}(e'_0)} = q$.

Now, we prove $q \leq \sizeexperimentbis{e}$. By Lemma~\ref{lemma:sizepoint atteint}, there exist $e' \sim e$ and a substitution $\sigma$ such that $\sizeexperiment{e'} = \sizepoint{\result{e'}} - \sizepoint{\mathcal{W}(e')}$, $\sigma(\result{e'}) = \result{e}$ and $\sigma(\mathcal{W}(e')) = \mathcal{W}(e)$. We have $q \leq \sizepoint{\result{e'}} - \sizepoint{\mathcal{W}(e')} + 2 \card{\mathcal{W}(e')} = \sizeexperiment{e'} + 2 \card{\mathcal{W}(e')} = \sizeexperimentbis{e'} = \sizeexperimentbis{e}$ (again by Fact~\ref{fact:equivExpMemeTaille}).
\end{proof}

\begin{proposition}\label{prop : sizebis of experimentbis of cut-free net}
Assume $A$ is infinite. Let $\pi$ be a cut-free net and let $(\mathbf{x}, \mathcal{V}) \in \smbis{\pi}$.

We have $\min \{ \sizeexperimentbis{e} \: ; \: e \textrm{ is an $\smbis{}$-experiment of $\pi$ such that $(\vert e \vert, \mathcal{W}(e)) = (\mathbf{x}, \mathcal{V})$} \} $ \\
$= \min \left\lbrace \sizepoint{\vert e' \vert} - \sizepoint{\mathcal{W}(e')} + 2 \card{\mathcal{W}(e')} \: ; \begin{array}{l} e' \textrm{ is an $\smbis{}$-experiment of $\pi$ such that} \\ (\exists \sigma \in \mathcal{S}) \: \sigma(\vert e' \vert, \mathcal{W}(e')) = (\mathbf{x}, \mathcal{V}) \end{array} \right\rbrace$.
\end{proposition}

\begin{proof}
Set $r = \min \left\lbrace \sizepoint{\vert e' \vert} - \sizepoint{\mathcal{W}(e')} + 2 \card{\mathcal{W}(e')} \: ; \begin{array}{l} e' \textrm{ is an $\smbis{}$-experiment of $\pi$ such that} \\ (\exists \sigma \in \mathcal{S}) \: \sigma(\vert e' \vert, \mathcal{W}(e')) = (\mathbf{x}, \mathcal{V}) \end{array} \right\rbrace$ and \\
$q = \min \{ \sizeexperimentbis{e} \: ; \: e \textrm{ is an $\smbis{}$-experiment of $\pi$ such that $(\vert e \vert, \mathcal{W}(e)) = (\mathbf{x}, \mathcal{V})$} \} $.

First we prove $q \leq r$. Let $e_0'$ be an $\smbis{}$-experiment of $\pi$ such that
\begin{itemize}
\item $(\exists \sigma \in \mathcal{S}) \: \sigma(\result{e_0'}, \mathcal{W}(e'_0)) = (\mathbf{x}, \mathcal{V})$
\item and $\sizepoint{\result{e'_0}} - \sizepoint{\mathcal{W}(e'_0)} + 2 \card{\mathcal{W}(e'_0)} = r$.
\end{itemize}
By Fact~\ref{fact:equivExpMemeTaille} and Lemma~\ref{lemma : closed by substitution}, there exists an $\smbis{}$-experiment $e_0$ of $\pi$ such that $\result{e_0} = (\mathbf{x}, \mathcal{V})$ and $\sizeexperimentbis{e_0} = \sizeexperimentbis{e'_0}$. By Lemma~\ref{lemma:sizeexperiment <= sizepoint}, we have $q \leq \sizeexperimentbis{e_0} = \sizeexperimentbis{e'_0} = \sizeexperiment{e'_0} + 2 \card{\mathcal{W}(e'_0)}
\leq \sizepoint{\result{e'_0}} - \sizepoint{\mathcal{W}(e'_0)} + 2 \card{\mathcal{W}(e'_0)} = r$.

Now, we prove $r \leq q$. Let $e$ be an $\smbis{}$-experiment of $\pi$ such that $\sizeexperimentbis{e} = q$ and $(\result{e}, \mathcal{W}(e)) = (\mathbf{x}, \mathcal{V})$. By Lemma~\ref{lemma:s(e)andequivalence}, we have $\sizeexperimentbis{e} = \min \{ \sizepoint{\result{e'}} - \sizepoint{\mathcal{W}(e')} + 2 \card{\mathcal{W}(e')} ; e' \sim e \textrm{ and } (\exists \sigma \in \mathcal{S}) \sigma(\result{e'}, \mathcal{W}(e')) = (\result{e}, \mathcal{W}(e)) \} \geq r$.
\end{proof}

We now state our main quantitative theorem, which answers Question~\ref{due2bis} raised in the introduction: using the notations of Corollary~\ref{corollary:postponingerasing}, we know that when $\cutnets{\pi}{\pi'}{c}{c'}$ is strongly normalizable, in order to compute $\strong{\cutnets{\pi}{\pi'}{c}{c'}}$ we have to compute the length of $R_{1}$ and $R_{2}$. We thus show how to compute $\length{R_1}$ and $\length{R_2}$ from $\smbis{\pi}$ and $\smbis{\pi'}$ (thus from $\sm{\pi}$ and $\sm{\pi'}$ thanks to Proposition~\ref{prop:SembisFromSem}).

\begin{theorem}\label{theorem:exactSN}
Assume $A$ is infinite. Let $\pi$ and $\pi'$ be two cut-free nets with conclusions $\textbf{d}, c$ (resp. $\textbf{d'}, c'$). The value of $\strong{\cutnets{\pi}{\pi'}{c}{c'}}$ is
$$\inf \left\lbrace \begin{array}{l} \frac{\sizepointbis{\mathbf{z}, \mathcal{W}} + \sizepointbis{\mathbf{z'}, \mathcal{W}'} - \sizebisinf{\smbis{\cutnets{\pi}{\pi'}{c}{c'}}}}{2} - \sizepoint{\mathcal{W} +  \mathcal{W}'} ; \\ (\mathbf{z}, \mathcal{W}) \in \smbis{\pi}, (\mathbf{z'}, \mathcal{W}') \in \smbis{\pi'} \textrm{ and } (\exists \sigma \in \mathcal{S}) \: \sigma(\mathbf{z}_c) = \sigma(\mathbf{z'}_{c'})^\perp \end{array} \right\rbrace$$
\end{theorem}

\begin{proof}
We set $$\mathcal{C} = \left\lbrace \begin{array}{l} \frac{\sizepointbis{\mathbf{z}, \mathcal{W}} + \sizepointbis{\mathbf{z'}, \mathcal{W}'} - \sizebisinf{\smbis{\cutnets{\pi}{\pi'}{c}{c'}}}}{2} - \sizepoint{\mathcal{W} +  \mathcal{W}'} ; \\ (\mathbf{z}, \mathcal{W}) \in \smbis{\pi}, (\mathbf{z'}, \mathcal{W}') \in \smbis{\pi'} \textrm{ and } (\exists \sigma \in \mathcal{S}) \: \sigma(\mathbf{z}_c) = \sigma(\mathbf{z'}_{c'})^\perp \end{array} \right\rbrace .$$

In the case where $\cutnets{\pi}{\pi'}{c}{c'}$ is not strongly normalizable, by Corollary~\ref{corollary : cut strongly normalizable} and Lemma~\ref{lemma : closed by substitution}, we have $\mathcal{C} = \emptyset$.

Now, we assume that $\cutnets{\pi}{\pi'}{c}{c'}$ is strongly normalizable. 

By Corollary~\ref{corollary:postponingerasing}, there exist $R_1 : \cutnets{\pi}{\pi'}{c}{c'} \stratnonerasingred \pi_1$ and $R_2 : \pi_1 \erasingred \pi_2$ antistratified such that
\begin{itemize}
\item $\pi_1$ is $\nonerasing$-normal;
\item $\pi_2$ is cut-free;
\item and $\strong{\cutnets{\pi}{\pi'}{c}{c'}} = \length{R_1} + \length{R_2}$.
\end{itemize}
By Corollary~\ref{corollary : cut strongly normalizable}, there are $(\mathbf{x}, \mathcal{V}) \in \smbis{\pi}$ and $(\mathbf{x'}, \mathcal{V'}) \in \smbis{\pi'}$ such that $\mathbf{x}_c = {\mathbf{x'}_{c'}}^\perp$: the set $\mathcal{C}$ is non-empty. 
We can thus consider some $(\mathbf{z}, \mathcal{W}) \in \smbis{\pi}, (\mathbf{z'}, \mathcal{W'}) \in \smbis{\pi'}$ and $\sigma \in \mathcal{S}$ such that $\sigma(\mathbf{z}_c) = \sigma(\mathbf{z'}_{c'})^\perp$ and $\frac{\sizepointbis{z, \mathcal{W}} + \sizepointbis{z', \mathcal{W'}} - \sizebisinf{\smbis{\cutnets{\pi}{\pi'}{c}{c'}}}}{2} - \sizepoint{\mathcal{W} +  \mathcal{W'}} = \min (\mathcal{C})$. We set $\mathbf{x} = \sigma(\mathbf{z})$, $\mathbf{x'} = \sigma(\mathbf{z'})$, $\mathcal{V} = \sigma(\mathcal{W})$ and $\mathcal{V'} = \sigma(\mathcal{W'})$. By Lemma~\ref{lemma : closed by substitution}, we have $(\mathbf{x}, \mathcal{V}) \in \smbis{\pi}$ and $(\mathbf{x'}, \mathcal{V'}) \in \smbis{\pi'}$. Since $\mathbf{x}_c = {\mathbf{x'}_{c'}}^\perp$, there exists a $\smbis{}$-experiment $e_0$ of $\cutnets{\pi}{\pi'}{c}{c'}$ such that 
\begin{itemize}
\item $\mathcal{W}(e_0) = \mathcal{V} + \mathcal{V'}$;
\item and 
\begin{eqnarray*}
\sizeexperimentbis{e_0} & = & \min \{ \sizeexperimentbis{e} \: ; \: e \textrm{ is an $\smbis{}$-experiment of $\pi$ such that }(\vert e \vert, \mathcal{W}(e)) = (\mathbf{x}, \mathcal{V}) \} + \\ & & \min \{ \sizeexperimentbis{e'} \: ; \: e' \textrm{ is an $\smbis{}$-experiment of $\pi'$ such that }(\vert e' \vert, \mathcal{W}(e')) = (\mathbf{x'}, \mathcal{V'}) \} .
\end{eqnarray*}
\end{itemize}
By applying Proposition~\ref{prop : sizebis of experimentbis of cut-free net} twice, we obtain 
\begin{eqnarray*}
\sizeexperimentbis{e_0} 
& = & \min \left\lbrace \sizepoint{\mathbf{z}} - \sizepoint{\mathcal{W}} + 2 \card{\mathcal{W}} \: ; \begin{array}{l} (\mathbf{z}, \mathcal{W}) \in \smbis{\pi} \textrm{ such that }\\ (\exists \sigma \in \mathcal{S}) \: (\sigma(\mathbf{z}), \sigma(\mathcal{W})) = (\mathbf{x}, \mathcal{V}) \end{array} \right\rbrace +\\
& & \min \left\lbrace \sizepoint{\mathbf{z'}} - \sizepoint{\mathcal{W'}} + 2 \card{\mathcal{W'}} \: ; \begin{array}{l} (\mathbf{z'}, \mathcal{W'}) \in \smbis{\pi'} \textrm{ such that }\\ (\exists \sigma \in \mathcal{S}) \: (\sigma(\mathbf{z'}), \sigma(\mathcal{W'})) = (\mathbf{x'}, \mathcal{V'}) \end{array} \right\rbrace \\
& = & \min \left\lbrace \begin{array}{l} \sizepoint{\mathbf{z}} -\sizepoint{\mathcal{W}} \\ + \sizepoint{\mathbf{z}'} - \sizepoint{\mathcal{W}'} \\ + 2 \card{\mathcal{W} + \mathcal{W'}} \end{array} ; \begin{array}{l}(\mathbf{z}, \mathcal{W}) \in \smbis{\pi}, (\mathbf{z'}, \mathcal{W'}) \in \smbis{\pi'} \textrm{ and } \\ \textrm{there exists } \sigma \in \mathcal{S} \textrm{ such} \\ \textrm{that }\sigma(\mathbf{z}, \mathcal{W}) = (\mathbf{x}, \mathcal{V}) \\ \textrm{and } \sigma(\mathbf{z'}, \mathcal{W'}) = (\mathbf{x'}, \mathcal{V'}) \end{array} \right\rbrace \\
& & \textrm{(the points of $\smbis{\pi}$ and $\smbis{\pi'}$ we look for are among those with disjoint atoms)}.
\end{eqnarray*}
Therefore we have $\sizeexperimentbis{e_0} \leq \sizepoint{\mathbf{z}} -\sizepoint{\mathcal{W}} + \sizepoint{\mathbf{z}'} - \sizepoint{\mathcal{W}'} + 2 \card{\mathcal{W} + \mathcal{W'}}$. Now, we have 
\begin{eqnarray*}
& & \sizepoint{\mathbf{z}} -\sizepoint{\mathcal{W}} + \sizepoint{\mathbf{z'}} - \sizepoint{\mathcal{W'}} + 2 \card{\mathcal{W} + \mathcal{W'}}\\
& = & 2 (\frac{\sizepointbis{\mathbf{z}, \mathcal{W}} + \sizepointbis{\mathbf{z'}, \mathcal{W'}} - \sizebisinf{\smbis{\cutnets{\pi}{\pi'}{c}{c'}}}}{2} - \sizepoint{\mathcal{W} +  \mathcal{W'}}) + \sizebisinf{\smbis{\cutnets{\pi}{\pi'}{c}{c'}}} \quad;
\end{eqnarray*}
remember that $\frac{\sizepointbis{z, \mathcal{W}} + \sizepointbis{z', \mathcal{W'}} - \sizebisinf{\smbis{\cutnets{\pi}{\pi'}{c}{c'}}}}{2} - \sizepoint{\mathcal{W} +  \mathcal{W'}} = \min (\mathcal{C})$, hence 
\begin{eqnarray*}
& & \sizepoint{\mathbf{z}} -\sizepoint{\mathcal{W}} + \sizepoint{\mathbf{z'}} - \sizepoint{\mathcal{W'}} + 2 \card{\mathcal{W} + \mathcal{W'}} \allowdisplaybreaks\\
& = & \min \left\lbrace \sizepoint{\mathbf{z}} -\sizepoint{\mathcal{W}} + \sizepoint{\mathbf{z'}} - \sizepoint{\mathcal{W'}} + 2 \card{\mathcal{W} + \mathcal{W'}} \: ; \: \begin{array}{c}(\mathbf{z}, \mathcal{W}) \in \smbis{\pi},\\ (\mathbf{z'}, \mathcal{W'}) \in \smbis{\pi'}\\ \textrm{ and } \\(\exists \sigma \in \mathcal{S}) \: \sigma(\mathbf{z}_c) = \sigma(\mathbf{z'}_{c'})^\perp \end{array} \right\rbrace \allowdisplaybreaks\\
& = & \min \left\lbrace \begin{array}{l} \sizepoint{\mathbf{z}} -\sizepoint{\mathcal{W}} \\ + \sizepoint{\mathbf{z'}} - \sizepoint{\mathcal{W'}} \\ + 2 \card{\mathcal{W} + \mathcal{W'}} \end{array} ; \begin{array}{l}(\mathbf{z}, \mathcal{W}) \in \smbis{\pi}, (\mathbf{z'}, \mathcal{W'}) \in \smbis{\pi'} \textrm{ and } \\ \textrm{there exist }(\mathbf{x}, \mathcal{V}), (\mathbf{x'}, \mathcal{V'}), \sigma \in \mathcal{S} \textrm{ such} \\ \textrm{that }\sigma(\mathbf{z}, \mathcal{W}) = (\mathbf{x}, \mathcal{V}), \sigma(\mathbf{z'}, \mathcal{W'}) = (\mathbf{x'}, \mathcal{V'}) \\\textrm{and } \mathbf{x}_c = \mathbf{x'}_{c'}^\perp \end{array} \right\rbrace \allowdisplaybreaks \\
& = & \min \left\lbrace \begin{array}{l} \sizepoint{\mathbf{z}} -\sizepoint{\mathcal{W}} \\ + \sizepoint{\mathbf{z'}} - \sizepoint{\mathcal{W'}} \\ + 2 \card{\mathcal{W} + \mathcal{W'}} \end{array} ; \begin{array}{l}(\mathbf{z}, \mathcal{W}) \in \smbis{\pi}, (\mathbf{z'}, \mathcal{W'}) \in \smbis{\pi'} \textrm{ and } \\ \textrm{there exist }(\mathbf{x}, \mathcal{V}) \in \smbis{\pi}, (\mathbf{x'}, \mathcal{V'}) \in \smbis{\pi'}, \sigma \in \mathcal{S} \textrm{ such} \\ \textrm{that }\sigma(\mathbf{z}, \mathcal{W}) = (\mathbf{x}, \mathcal{V}), \sigma(\mathbf{z'}, \mathcal{W'}) = (\mathbf{x'}, \mathcal{V'}) \\\textrm{and } \mathbf{x}_c = \mathbf{x'}_{c'}^\perp \end{array} \right\rbrace \\
& & \textrm{(by Lemma~\ref{lemma : closed by substitution})} \allowdisplaybreaks \\
& = & \min \left\lbrace  \begin{array}{l} \min \{ \sizeexperimentbis{e} \: ; \: e \textrm{ is an $\smbis{}$-experiment of $\pi$ such that }(\vert e \vert, \mathcal{W}(e)) = (\mathbf{x}, \mathcal{V}) \} + \\ \min \{ \sizeexperimentbis{e'} \: ; \: e' \textrm{ is an $\smbis{}$-experiment of $\pi'$ such that }(\vert e' \vert, \mathcal{W}(e')) = (\mathbf{x'}, \mathcal{V'}) \} ; \\ (\mathbf{x}, \mathcal{V}) \in \smbis{\pi}, (\mathbf{x'}, \mathcal{V'}) \in \smbis{\pi'} \textrm{ and } \mathbf{x}_c = \mathbf{x'}_{c'}^\perp \end{array} \right\rbrace \\
& & \textrm{(by applying Proposition~\ref{prop : sizebis of experimentbis of cut-free net} twice)} \allowdisplaybreaks\\
& = & \min \left\lbrace \sizeexperimentbis{e} + \sizeexperimentbis{e'} \: ; \begin{array}{l} 
e \textrm{ is an $\smbis{}$-experiment of $\pi$,} 
e' \textrm{ is an $\smbis{}$-experiment of $\pi'$} \\
\textrm{and } (\exists (\mathbf{x}, \mathcal{V}) \in \smbis{\pi}, (\mathbf{x'}, \mathcal{V'}) \in \smbis{\pi'}) \\ ((\vert e \vert, \mathcal{W}(e)) = (\mathbf{x}, \mathcal{V}) \textrm{ and } (\vert e' \vert, \mathcal{W}(e')) = (\mathbf{x'}, \mathcal{V'}) \textrm{ and } \mathbf{x}_c = {\mathbf{x'}_{c'}}^\perp) \end{array} \right\rbrace \allowdisplaybreaks\\
& = & \min \{ \sizeexperimentbis{e} \: ; \: e \textrm{ is an $\smbis{}$-experiment of } \cutnets{\pi}{\pi'}{c}{c'} \} \leq \sizeexperimentbis{e_0} .
\end{eqnarray*}

So, $\sizeexperimentbis{e_0} = \sizepoint{\mathbf{z}} -\sizepoint{\mathcal{W}} + \sizepoint{\mathbf{z'}} - \sizepoint{\mathcal{W'}} + 2 \card{\mathcal{W} + \mathcal{W'}}=$\\ 
$= \min \{ \sizeexperimentbis{e} \: ; \: e \textrm{ is an $\smbis{}$-experiment of } \cutnets{\pi}{\pi'}{c}{c'} \} .$ 
Since $\mathcal{W}(e_0) = \mathcal{V} + \mathcal{V'}$ and $\card{\mathcal{V} + \mathcal{V'}} = \card{\mathcal{W} + \mathcal{W'}}$, we have $\sizeexperiment{e_0} = \sizeexperimentbis{e_0} - 2 \card{\mathcal{W} + \mathcal{W'}} = \sizepoint{\mathbf{z}} -\sizepoint{\mathcal{W}} + \sizepoint{\mathbf{z'}} - \sizepoint{\mathcal{W'}}$.

By Proposition~\ref{proposition : non-erasing stratified reduction}, we have
$\length{R_1} = (\sizeexperiment{e_0} - \sizebisinf{\smbis{\cutnets{\pi}{\pi'}{c}{c'}}})/2 = 
(\sizepoint{\mathbf{z}} -\sizepoint{\mathcal{W}} + \sizepoint{\mathbf{z'}} - \sizepoint{\mathcal{W'}} - \sizebisinf{\smbis{\cutnets{\pi}{\pi'}{c}{c'}}})/2$. Moreover, by Lemma~\ref{lemma : key-lemma : strat}, there exists $e_1 :_{\smbis{}} \pi_1$ s.t.
\begin{itemize}
\item $\weakeningsofexperiment{e_1} = \mathcal{V} + \mathcal{V'}$ 
\item and $\sizeexperimentbis{e_1} = \min \{ \sizeexperimentbis{e} \: ; \: e \textrm{ is an $\smbis{}$-experiment of } \pi_1 \}$.
\end{itemize}
By Lemma~\ref{lemma : erasing cuts} (applied to $\pi_{1}$ and $e_{1}$), there are $\card{\mathcal{V} + \mathcal{V'}} = \card{\mathcal{W} + \mathcal{W'}}$ (erasing) cuts in $\pi_{1}$. Since $\pi_1$ is $\nonerasing$-normal and $R_{2}$ is antistratified, we have $\length{R_2} = \card{\mathcal{W} + \mathcal{W'}}$. Hence
\begin{eqnarray*}
& & \strong{\cutnets{\pi}{\pi'}{c}{c'}} \\
& = & \length{R_1} + \length{R_2} \allowdisplaybreaks\\
& = & (\sizepoint{\mathbf{z}} -\sizepoint{\mathcal{W}} + \sizepoint{\mathbf{z'}} - \sizepoint{\mathcal{W'}} +2 \card{\mathcal{W} + \mathcal{W'}} - \sizebisinf{\smbis{\cutnets{\pi}{\pi'}{c}{c'}}})/2   \allowdisplaybreaks\\
& = & \frac{\sizepoint{\mathbf{z}} + \sizepoint{\mathcal{W}} + 2 \card{\mathcal{W}} + \sizepoint{\mathbf{z'}} + \sizepoint{\mathcal{W'}} + 2 \card{\mathcal{W'}} - \sizebisinf{\smbis{\cutnets{\pi}{\pi'}{c}{c'}}}}{2} \\
& & - \sizepoint{\mathcal{W} +  \mathcal{W'}} \allowdisplaybreaks\\
& = & \frac{\sizepointbis{\mathbf{z}, \mathcal{W}} + \sizepointbis{\mathbf{z'}, \mathcal{W'}} - \sizebisinf{\smbis{\cutnets{\pi}{\pi'}{c}{c'}}}}{2} - \sizepoint{\mathcal{W} +  \mathcal{W'}} \allowdisplaybreaks\\
& = & \min (\mathcal{C}) .
\end{eqnarray*}
\end{proof}

We now give a concrete example of application of Theorem~\ref{theorem:exactSN}, which has also a theoretical purpose: we want to show that only a little part of $\sm{\pi}$ and $\sm{\pi'}$ is involved in the computation of $\strong{\cutnets{\pi}{\pi'}{c}{c'}}$. In~\cite{MR2926280}, we proved that from $\sm{\pi}$ one can recover much information about $\pi$ (the whole net $\pi$ in the absence of weakenings). And when it is possible to recover $\pi$ from $\sm{\pi}$, a straightforward way to compute $\strong{\cutnets{\pi}{\pi'}{c}{c'}}$ from $\sm{\pi}$ and $\sm{\pi'}$ is to recover $\pi$ and $\pi'$ from $\sm{\pi}$ and $\sm{\pi'}$, and then to apply the cut elimination procedure to the net $\cutnets{\pi}{\pi'}{c}{c'}$! Of course this is not at all what Theorem~\ref{theorem:exactSN} does, and to illustrate this fact, we consider a net $\pi$, two nets $\pi'_1$, $\pi'_2$ with the same conclusions (represented in Figure~\ref{figure:example}) and the two nets $\cutnets{\pi}{\pi'_{1}}{c}{c'}$ and $\cutnets{\pi}{\pi'_{2}}{c}{c'}$. As noticed by Pierre Boudes (see \cite{Boudes-Desequentialized} for a formulation in the framework of Abstract B\"ohm trees), the elements of $\sm{\pi'_1}$ and of $\sm{\pi'_2}$ in which the positive multisets have cardinality $0$ or $1$ are the same (which entails that these points are not enough to recover $\pi'_{1}$ nor $\pi'_{2}$ since clearly $\pi'_{1}\neq\pi'_{2}$). We show here that these points are nevertheless enough to compute $\strong{\cutnets{\pi}{\pi'_1}{c}{c'}}$ and $\strong{\cutnets{\pi}{\pi'_2}{c}{c'}}$, following the method proposed in Theorem~\ref{theorem:exactSN}. This clearly shows that the amount of information required to apply our method is much less than the one required to recover the nets themselves, which is desirable, since the information we obtain (the maximal length of the reduction sequences) is itself less that the complete computation.

\begin{example}\label{example:LessThanInjectivity}
Let $\pi$ (resp. $\pi'_1$, $\pi'_2$) be the net of Figure~\ref{figure:example} with conclusions $d, c$ (resp. $c'$). 
\begin{figure}
\begin{center}
\scalebox{\scalefact}{\input{Boudes.pstex_t}}
\caption{An example}\label{figure:example}
\end{center}
\end{figure}
Notice that we have\\ 
$\langle (-, [(+, \ast), (+, \ast)]), (+, [(-, [(+, [(-, \ast)])]), (-, [(+, [(-, \ast)])])]) \rangle \in \sm{\pi}$ and\\
 $\langle (-, [(+, [(-, [(+, \ast)])]), (+, [(-, [])])]) \rangle \in \sm{\pi'_1}, \sm{\pi'_2}$. 
We thus have, by Proposition~\ref{prop:SembisFromSem}, 
\begin{itemize}
\item 
$(\langle (-, [(+, \ast), (+, \ast)]), (+, [(-, [(+, [(-, \ast)])]), (-, [(+, [(-, \ast)])])]) \rangle, []) \in \smbis{\pi}$ (indeed, since the considered point of $\sm{\pi}$ is exhaustive and does not contain $(-,[])$, intuitively it is also a point of $\smbis{\pi}$)
\item and $(\langle (-, [(+, [(-, [(+, \ast)])]), (+, [(-, [(+, \ast)])])]) \rangle, [(+, \ast)]) \in \smbis{\pi'_1}, \smbis{\pi'_2}$ (here, contrary to the previous case, the function $F$ of Definition~\ref{def:F} really plays a role: we have considered the point obtained by substituting $(-, [])$ with $(-, [(+, \ast)])$, where $(+, \ast)\in D$).
\end{itemize}
We have 
\begin{itemize}
\item $\sizebisinf{\smbis{\cutnets{\pi}{\pi'_1}{c}{c'}}} = \sizebisinf{\smbis{\cutnets{\pi}{\pi'_2}{c}{c'}}} = \sizepointbis{(\langle (-, [(+, \ast), (+, \ast)]) \rangle,
[(+, \ast)])} =\sizepoint{\langle (-, [(+, \ast), (+, \ast)]) \rangle}+(\sizepoint{(+, \ast)}+2) =3+(1+2)=6$, 
\item $\sizepointbis{(\langle (-, [(+, \ast), (+, \ast)]), (+, [(-, [(+, [(-, \ast)])]), (-, [(+, [(-, \ast)])])]) \rangle, [])} = \sizepoint{\langle (-, [(+, \ast), (+, \ast)]), (+, [(-, [(+, [(-, \ast)])]), (-, [(+, [(-, \ast)])])]) \rangle}=10$,
\item $\sizepointbis{(\langle (-, [(+, [(-, [(+, \ast)])]), (+, [(-, [(+, \ast)])])]) \rangle, [(+, \ast)])}$ \\
$= \sizepoint{\langle (-, [(+, [(-, [(+, \ast)])]), (+, [(-, [(+, \ast)])])]) \rangle}+(\sizepoint{(+, \ast)}+2)=7+(1+2)=10$,
\end{itemize}
hence, by Theorem~\ref{theorem:exactSN}, we have $\strong{\cutnets{\pi}{\pi'_1}{c}{c'}}, \strong{\cutnets{\pi}{\pi'_2}{c}{c'}} \leq \frac{10+10-6}{2}-1 = 6$. Actually, one can check that these points are those which give the exact value of $\strong{\cutnets{\pi}{\pi'_1}{c}{c'}}$ and of $\strong{\cutnets{\pi}{\pi'_2}{c}{c'}}$: $\strong{\cutnets{\pi}{\pi'_1}{c}{c'}} = \strong{\cutnets{\pi}{\pi'_2}{c}{c'}} = 6$.
\end{example}

\section*{Conclusion}

We introduced a new interpretation $\smbis{-}$ of nets and showed that, for any net $\pi$, we have $\smbis{\pi} \not= \emptyset$ if, and only if, $\pi$ is strongly normalizing. In order to prove this theorem, we showed by the way, without using confluence, the Conservation Theorem ($\textbf{WN}^{\nonerasing} = \textbf{SN}$) - a key point in several proofs of strong normalization. 

This characterization of strong normalization has been refined with quantitative information relating the exact number of reduction steps of longest reduction sequences and some size of $\smbis{}$-experiments. This relation applied to the case of a net consisting of the cut of two cut-free nets allowed to show that the size of some well-chosen points gives the exact number of reduction steps of longest reduction sequences, even if these points are not enough to reconstruct the net.

Of course, the $\smbis{}$-interpretation does not provide a denotational semantics in that this interpretation is not invariant during the reduction. This new interpretation is actually a variant of a well-known interpretation, the $\sm{}$-interpretation, which does provide a denotational semantics: given the $\sm{}$-interpretation of a cut-free $\pi$, we can compute its $\smbis{}$-interpretation, even wihout reconstructing the net (unlike with $\lambda$-terms, it is not always possible to reconstruct a net from its interpretation in some denotational semantics, and even if it is possible, it is generally very difficult and not trivial at all). The $\smbis{}$-interpretation, when restricted to nets corresponding to $\lambda$-terms, corresponds to some non-idempotent intersection types system, called here System~$R^\textit{ex}$, in the same way as the $\sm{}$-interpetation corresponds to the non-idempotent intersection types system called System~R. System~$R^\textit{ex}$ is very close to the system studied in \cite{bernadetlengrand13} , which identified a measure on typing derivations that, for some specific derivations, provides the exact number of measure of longest reduction sequences of $\beta$-reduction steps, while a similar work was done for System~R and steps of Krivine's machine in \cite{phddecarvalho, Carvalhoexecution}.

Since we showed that only a small part of the semantics is used to determine the number of reduction steps (a small part which -in general- is not enough to recover the syntax), an interesting problem is to know whether we could obtain a similar result using the multiset based coherence semantics, for which we know since~\cite{phdtortora} that it is in general impossible to recover a net from its interpretation.

\paragraph{Acknowledgements.} We thank Alexis Bernadet and St\'ephane Graham-Lengrand for a stimulating discussion on the previous version of this work.

\bibliographystyle{plain}
\bibliography{ll}

\end{document}